%% file: HP_DynamicPartioning_v5_ForArXiv.tex
\documentclass[onecolumn,draftclsnofoot,12pt]{IEEEtran}
\input{input.tex}
\allowdisplaybreaks

\usepackage{epstopdf}
\usepackage{enumerate}
\usepackage{algorithmicx} 
\usepackage{algorithm}
\usepackage{amsmath}
\usepackage[noend]{algpseudocode}
\usepackage{float}
\usepackage{color}
\usepackage{makeidx}
\usepackage{bbm}
\newcommand{\sref}[1]{{Section}~\ref{#1}}

\def\Fbb{\bF_\mathrm{BB}}
\def\Frf{\bF_\mathrm{RF}}

\def \Nr {N_{\rm{RX}}}
\def \Nrf {N_{\rm{RF}}}
\def \Nt {N_{\rm{TX}}}
\def \Frf {\bF_{\rm{RF}}}
\def \Fbb {\bF_{\rm{BB}}}

\begin{document}
\title{Dynamic Subarrays  for Hybrid Precoding in Wideband mmWave MIMO Systems}
\author{ Sungwoo Park, Ahmed Alkhateeb, and Robert W. Heath Jr. \thanks{Sungwoo Park, Ahmed Alkhateeb, and Robert W. Heath Jr. are with The University of Texas at Austin (Email:  spark, aalkhateeb, rheath@utexas.edu).} \thanks{This work is supported in part by the National Science Foundation under Grant No. 1319556, and by a gift from Huawei Technologies.}}
\maketitle 

\begin{abstract}
	Hybrid analog/digital precoding architectures can address the trade-off between achievable spectral efficiency and power consumption in large-scale MIMO systems. This makes it a promising candidate for millimeter wave systems, which require deploying large antenna arrays at both the transmitter and receiver to guarantee sufficient received signal power. Most prior work on hybrid precoding focused on narrowband channels and assumed fully-connected hybrid architectures. MmWave systems, though, are expected to be wideband with frequency selectivity. In this paper, a closed-form solution for fully-connected OFDM-based hybrid analog/digital precoding is developed for frequency selective mmWave systems. This solution is then extended to partially-connected but fixed architectures in which each RF chain is connected to a specific subset of the antennas. The derived solutions give insights into how the hybrid subarray structures should be designed. Based on them, a novel technique that dynamically constructs the hybrid subarrays based on the long-term channel characteristics is developed. Simulation results show that the proposed hybrid precoding solutions achieve spectral efficiencies close to that obtained with fully-digital architectures in wideband mmWave channels.  Further, the results indicate that the developed dynamic subarray solution outperforms the fixed hybrid subarray structures in various system and channel conditions.  
\end{abstract}

%%%%%%%%%%%%%%%%%%%%%%%%%%%%%%%%%%%%%%%%%%%%%%%%%%%%%%%%%%%%%%%%%%%%%%%%%%%%%%%%%%%%%%%%%%%%%%%%%%%%%%%%%%%%%
\section{Introduction} \label{sec:Intro}
%%%%%%%%%%%%%%%%%%%%%%%%%%%%%%%%%%%%%%%%%%%%%%%%%%%%%%%%%%%%%%%%%%%%%%%%%%%%%%%%%%%%%%%%%%%%%%%%%%%%%%%%%%%%%

% Motivation 

Hybrid analog/digital architectures are efficient transceivers for millimeter wave (mmWave) systems \cite{Zhang2005a,Venkateswaran2010,ElAyach2014,Alkhateeb2014d,Roh2014,HeathJr2015}. These architectures enable a flexible compromise between achieving high spectral efficiency and maintaining low cost and power consumption. Extensive work has been devoted to developing hybrid precoding algorithms to single-user and multi-user mmWave and massive MIMO systems in the last few years \cite{ElAyach2014,Alkhateeb2013,Ni2016,Yu2016,Mendez-Rial2015a,Chen2015,Han2015,ElAyach2013}. Most prior work on hybrid precoding focused on narrowband channels. MmWave systems, however, will likely operate on wideband channels with frequency selectivity \cite{Rappaport2013a,Samimi2014,Pi2011}. It is, therefore, important to develop hybrid analog/digital precoding designs for frequency selective mmWave systems. 

\subsection{Prior Work}

 Hybrid architectures divide the processing needed for precoding and combining between analog and digital domains to reduce the number of RF chains \cite{Zhang2005a,Venkateswaran2010,ElAyach2014,Alkhateeb2013,Ni2016,Yu2016,Mendez-Rial2015a,Chen2015,Han2015,ElAyach2013,Kim2013,Alkhateeb2015a}. In \cite{Zhang2005a,Venkateswaran2010}, hybrid precoding was first investigated for diversity and  multiplexing gains in general MIMO systems. For mmWave large MIMO systems, \cite{ElAyach2014} leveraged the sparse nature of mmWave channels and designed low-complexity hybrid precoding algorithms based on orthogonal matching pursuit. Following \cite{ElAyach2014}, the work in \cite{Alkhateeb2013,Ni2016,Yu2016,Mendez-Rial2015a,Chen2015} devised hybrid precoding algorithms based on matrix decomposition, alternative minimization, and other techniques, with the objective of achieving spectral efficiencies close to that obtained with fully-digital solutions. The system models in \cite{ElAyach2014,Alkhateeb2013,Ni2016,Yu2016,Mendez-Rial2015a,Chen2015} adopted a fully-connected hybrid architecture, meaning that each RF chain is connected to all the antennas. Extensions to subarray-based hybrid architectures
were considered in \cite{Han2015,ElAyach2013}. The work in \cite{ElAyach2014,Alkhateeb2013,Ni2016,Yu2016,Mendez-Rial2015a,Chen2015,Han2015,ElAyach2013}, though, assumed a narrowband mmWave channel, with perfect or partial channel knowledge at the transmitter.

Limited work has been done for wideband mmWave hybrid precoding systems. In \cite{Kim2013}, hybrid beamforming with only a single-stream transmission over MIMO-OFDM systems was considered. The developed solution in \cite{Kim2013}, though, relied on the exhaustive search over the RF and baseband codebooks, and did not provide specific criteria for the design of these codebooks. In \cite{Alkhateeb2015a}, for OFDM-based mmWave hybrid precoding systems, the optimal baseband precoders for a given RF codebook were obtained, and efficient codebooks were designed. The work in \cite{Alkhateeb2015a}, however, did not exploit the channel correlation over the adjacent subcarriers to reduce the precoder design complexity. Further, the work in \cite{Kim2013,Alkhateeb2015a} considered only the fully-connected hybrid architecture, which consumes more power consumption compared to the subarray structure \cite{Han2015,ElAyach2013}, which connects each RF chain to only a subset of the antennas.

\subsection{Contribution}

In this paper, we develop hybrid precoding designs for wideband mmWave large MIMO systems. The contributions of this paper are summarized as follows.
\begin{itemize} 
\item {We develop a near-optimal closed-form solution for fully-connected and partially-connected hybrid analog/digital precoding in OFDM-based wideband mmWave systems. In our design, we assume fully-digital receivers and adopt a relaxation of the original mutual information maximization problem. For the relaxed problem, we obtain the optimal baseband and RF precoders. The developed solution has exactly the same spectral efficiency as the unconstrained fully-digital solution if the number of channel paths is less than the number of RF chains. Therefore, thanks to the sparse nature of the mmWave channel, the proposed hybrid precoding with a small number of RF chains can achieve a spectral efficiency near to that obtained with the unconstrained fully-digitalized baseband precoding. Further, the developed closed-form solution provides insights into the impact of the subarray structures on the overall system performance.}

\item {We propose a criterion to construct the optimal subarrays that maximize a proxy of the system spectral efficiency, i.e., the best partitioning/grouping of the antennas over the RF chains. 
Using this criterion, we propose a dynamic subarray structure that adapts the subarray structure according to the long-term channel statistics. %, such as the spatial channel covariance. 
Finding the optimal subarrays requires an exhaustive search over many antenna partitioning solutions. To lower the complexity, we propose a greedy algorithm that approaches the spectral efficiency of the optimal exhaustive search solution. }

%Even for a fixed subarray structure, the proposed criterion can be useful for choosing the best subarray structure. Consider an example of a uniform planar array (UPA) in a 3D-MIMO channel. The proposed criterion indicates that forming each subarray with vertical antenna elements is a better choice than forming one with horizontal antenna elements if the elevation angles of arrival and departure are concentrated in a smaller range while the azimuth angles are not. If both the azimuth angles and the elevation angles have the similar ranges, then the squared subarray type in which each subarray is composed of adjacent antenna elements in both dimensions outperforms the vertical or horizontal subarray type.   }
\end{itemize}

The proposed hybrid precoding designs were also evaluated by simulations. Results show that the developed wideband hybrid precoding design approaches the spectral efficiencies of the fully-digital solutions for both fully-connected and fixed-subarray architectures. For the dynamic subarrays, results indicate that their performance outperforms any fixed subarray structure, promoting their potential advantages in wideband mmWave systems. 

% Notation
\noindent \textbf{Notation:} We use the following notation throughout this paper: $\bA$ is a matrix, $\ba$ is a vector, $a$ is a scalar, and $\cA$ is a set. $|a|$ and $\measuredangle{a}$ are the magnitude and phase of the complex number $a$. $\|\bA \|_F$ is its Frobenius norm, and $\bA^T$, $\bA^*$, and $\bA^{-1}$ are its transpose, Hermitian (conjugate transpose), and inverse, respectively. $[\bA]_{1:k}$ denotes the matrix that is composed of the first $k$ columns of the matrix $\bA$. $\mathrm{diag}(\ba)$ is a diagonal matrix with the entries of $\ba$ on its diagonal, and $\mathrm{blkdiag} \left(\ba_1, \cdots, \ba_k \right)$ is a block diagonal matrix with $\ba_i$'s on its diagonal blocks. $[\bA]_{m,n}$ is the $(m,n)$-th element of the matrix $\bA$. $\measuredangle{\bA}$ is a matrix with the $(m,n)$-th element equals $e^{j \left[\bA \right]_{m,n}}$. $\bI$ is the identity matrix and $\mathbf{1}_{N}$ is the $N$-dimensional all-ones vector. $\mathcal{CN}(\bm,\bR)$ is a complex Gaussian random vector with mean $\bm$ and covariance $\bR$. $\bbE\left[\cdot\right]$ is used to denote expectation.
%$\bA \circ \bB$ is the Khatri-Rao product of $\bA$, and $\bB$, $ \bA \otimes \bB$ is the Kronecker product of $\bA$, and $\bB$, and

%%%%%%%%%%%%%%%%%%%%%%%%%%%%%%%%%%%%%%%%%%%%%%%%%%%%%%%%%%%%%%%%%%%%%%%%%%%%%%%%%%%%%%%%%%%%%%%%%%%%%%%%%%%%%
\section{System and Channel Models} \label{sec:Model}
%%%%%%%%%%%%%%%%%%%%%%%%%%%%%%%%%%%%%%%%%%%%%%%%%%%%%%%%%%%%%%%%%%%%%%%%%%%%%%%%%%%%%%%%%%%%%%%%%%%%%%%%%%%%%
In this section, we introduce the adopted system and channel models for wideband hybrid precoding.
%--------------------------------------------
\subsection{System Model}
%--------------------------------------------
\begin{figure}[!t]
	\centerline{\resizebox{1.0\columnwidth}{!}{\includegraphics{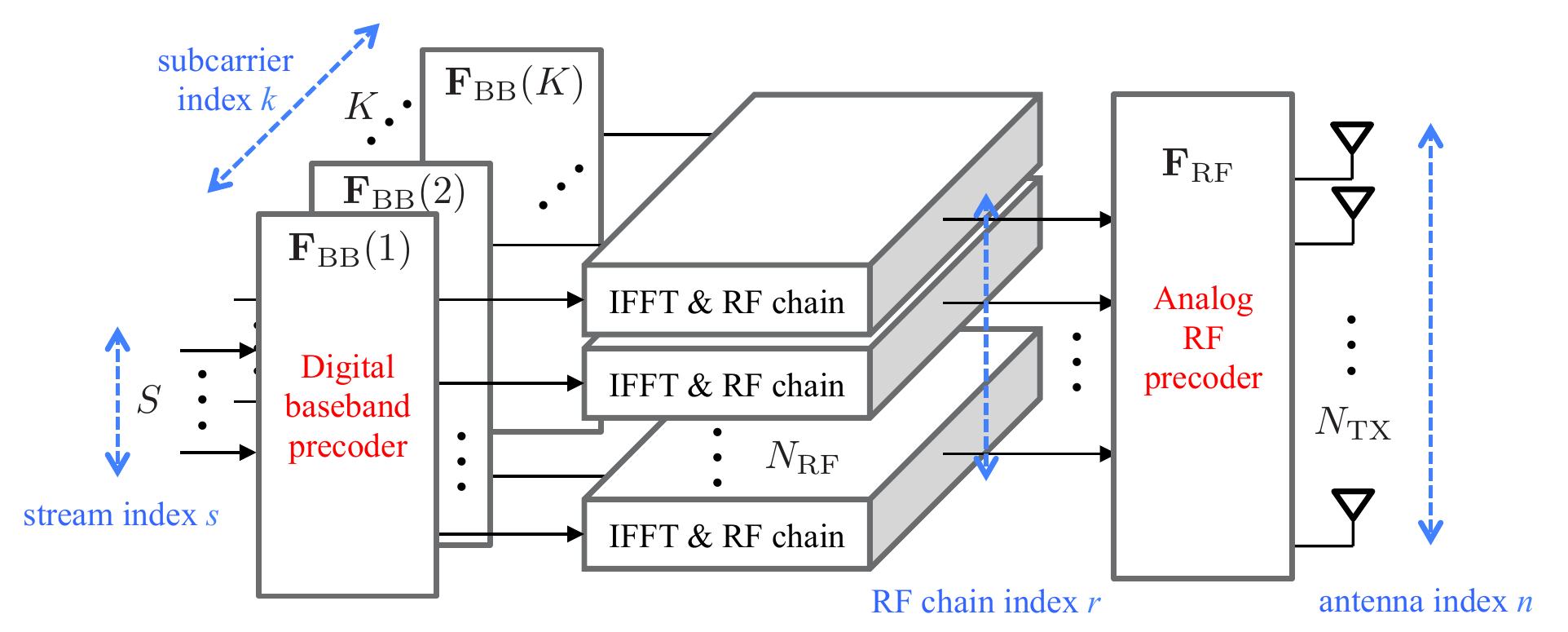}}}
	%\centering
	%\includegraphics[width=3.7in, height=3.1in]{Fig2.pdf}
	\caption{Hybrid precoding architecture in a wideband OFDM-MIMO system.}\label{fig:WB_HP_architecture}
\end{figure}
Consider the system model in \figref{fig:WB_HP_architecture}. A base station (BS) with $\Nt$ antennas and $\Nrf$ RF chains, $\Nrf \leq \Nt$, communicates with a mobile user that has $\Nr$ antennas, via $S$ streams. We assume in this paper that the number of RF chains at the mobile user is equal to the number of antennas, and focus on the hybrid precoding design at the BS. We adopt a hybrid precoding MIMO-OFDM transmission model, similar to \cite{Alkhateeb2015a} with $K$ subcarriers. Let  $ \bF_{\mathrm{RF}}$ be an $\Nt \times \Nrf$ wideband analog RF precoding matrix, and $\bF_{\rm{BB}}[k]$ be an $\Nrf \times S$ matrix that represents the digital baseband precoding at the $k$-th subcarrier.  The transmitted signal can be expressed as
\begin{equation}\label{eq:WB_hybrid_precoding_model}
\bx[k] = \bF_{\rm{RF}}  \bF_{\rm{BB}}[k] \bs[k], \;\; \textrm{for} \;\; k=1,\dots,K,
\end{equation}
where $\bs[k]$ is the $S \times 1$ vector of transmitted symbols at subcarrier $k$ with $\bbE\left[\bs[k] \bs^*[k]\right]=\bI_{S}$.  While the analog RF precoding, $ \bF_{\rm{RF}} $, is performed in the time domain and the same precoding matrix is applied for the entire bandwidth, the digital baseband precoding, $ \bF_{\rm{BB}}[k] $, is performed in the frequency domain on a per-subcarrier basis. This is a main distinguishing feature of OFDM-based hybrid precoding compared with fully-digital precoding. The precoders,  $ \bF_{\rm{RF}} $ and  $ \bF_{\rm{BB}}[k] $, are coupled together through the total power constraints, $\sum_{k=1}^{K} || \bF_{\rm{RF}} \bF_{\rm{BB}}[k]||_{\rm{F}}^{2} \leq P_{\rm{tot}} $, where $P_\mathrm{tot}$ is the total transmit power. 

At the receiver, assuming perfect carrier and frequency offset synchronization, the cyclic prefix of length $D$ is first removed from the received signal. The symbols at each subcarrier $k$ are then combined using the $\Nr \times N_\mathrm{S}$ digital combining matrix $\bW[k]$. Note that no hybrid combining is assumed as the number of RF chains at the receiver equals the number of antennas. Let the  $N_\mathrm{RX} \times N_\mathrm{TX}$ matrix $\bH[k]$ denote the channel at subcarrier $k$, the received signal at subcarrier $k$ after processing can be then written as
\begin{equation}
\by[k]=\bW^*[k] \bH[k] \Frf \Fbb[k] s[k]+ \bW^*[k] \bn[k],
\label{eq:processed}
\end{equation}
where  $\bn[k] \sim \mathcal{CN}(\boldsymbol{0}, \sigma_\mathrm{N}^2 \bI)$ is the Gaussian noise vector at the receiver.  
%----------------------------------------

\subsection{Channel Model}
%----------------------------------------
We adopt a geometric channel model to incorporate the wideband and limited scattering characteristics of mmWave channels \cite{Rappaport2013a,Samimi2014,Rappaport2012}. Consider a mmWave channel with $N_{\rm{CH}}$ paths between the BS and mobile user, and let $\rho_{\ell}$, $\tau_\ell$, $\phi_{{\rm{R}},\ell}$, $\theta_{{\rm{R}},\ell}$, $\phi_{{\rm{T}},\ell}$, $\theta_{{\rm{T}},\ell}$ denote the $\ell$th complex path gain, delay, azimuth angle of departure, elevation angle of departure, azimuth angle of arrival, and elevation angle of arrival,  respectively. Let $p(\tau)$ denote a pulse shaping filter for $T_{\rm{s}}$-spaced signaling at $\tau$ seconds. The delay-$d$ MIMO channel matrix can be written as \cite{SchniterSayeed2014,Alkhateeb2015a}
\begin{equation}\label{eq:CIR_model}
\bH[d] = \sum_{p=1}^{N_{\rm{CH}}} \alpha_{p} p( d T_{\rm{s}} - \tau_p ) \ba_{{\rm{R}}} ( \phi_{{\rm{R}},p}, \theta_{{\rm{R}},p})  \ba^{*}_{{\rm{T}}} ( \phi_{{\rm{T}},p}, \theta_{{\rm{T}},p} ), 
\end{equation}
where $\ba_{{\rm{T}}} \left( \phi_{{\rm{T}},p}, \theta_{{\rm{T}},p} \right)$ and $\ba_{{\rm{R}}} \left( \phi_{{\rm{R}},p}, \theta_{{\rm{R}},p} \right)$ represent the transmit and receive array response vectors, which depend on the antenna array type. 
Assuming perfect synchronization, the channel frequency response matrix at each subcarrier $k$ can be expressed as 
\begin{equation}\label{eq:H_k}
\bH[k] = \sum_{p=1}^{N_{\rm{CH}}}  \alpha_{p}  \omega_{\tau_p}[k]  \ba_{{\rm{R}}} ( \phi_{{\rm{R}},p}, \theta_{{\rm{R}},p} )  \ba^{*}_{{\rm{T}}} ( \phi_{{\rm{T}},p}, \theta_{{\rm{T}},p} ),
\end{equation}
%where $\omega_{\tau_p}[k] = \sum_{d=0}^{D-1}  p( d T_{\rm{s}} - \tau_p ) e^{-\frac{j2\pi k  d }{K}}$.
where $\omega_{\tau_p}[k] $ is defined as
\begin{equation}\label{eq:omega_def}
\omega_{\tau_p}[k] = \sum_{d=0}^{D-1}  p( d T_{\rm{s}} - \tau_p ) e^{-\frac{j2\pi k  d }{K}}.
\end{equation}
The channel matrix $\bH[k]$ can also be written in a more compact form as
\begin{equation}\label{eq:H_k_matrix}
\bH[k]= \bA_{{\rm{R}}}  \bD[k] \bA^*_{{\rm{T}}},
\end{equation}
where $\bA_\mathrm{R}$ and $\bA_\mathrm{T}$ carry the array response vectors of the transmitter and receiver as 
\begin{equation}\label{eq:A_R_A_T_matrix}
\begin{split}
\bA_{{\rm{R}}}  &=  \begin{bmatrix} \ba_{{\rm{R}}} ( \phi_{{\rm{R}},1}, \theta_{{\rm{R}},1} ) &  \ba_{{\rm{R}}}  (\phi_{{\rm{R}},2}, \theta_{{\rm{R}},2} ) & \cdots & \ba_{{\rm{R}}} ( \phi_{{\rm{R}},N_{\rm{CH}}}, \theta_{{\rm{R}},N_{\rm{CH}}} ) \end{bmatrix} \\
\bA_{{\rm{T}}} &=  \begin{bmatrix} \ba_{{\rm{T}}} ( \phi_{{\rm{T}},1}, \theta_{{\rm{T}},1} ) &  \ba_{{\rm{T}}}  (\phi_{{\rm{T}},2}, \theta_{{\rm{R}},2} ) & \cdots & \ba_{{\rm{T}}} ( \phi_{{\rm{T}},N_{\rm{CH}}}, \theta_{{\rm{T}},N_{\rm{CH}}} ) \end{bmatrix}, 
\end{split}
\end{equation}
and the diagonal matrix $\bD[k]$ equals 
%\begin{equation}\label{eq:W_hat_k}
%\bD[k] =  \mathrm{diag} \left(\begin{bmatrix}  \alpha_{1} \sum_{d=0}^{D-1}  p( d T_{\rm{s}} - \tau_1 ) e^{-\frac{j2\pi k  d }{K}}  & \cdots &  \alpha_{N_{\rm{CH}}} \sum_{d=0}^{D-1}  p( d T_{\rm{s}} - \tau_{N_{\rm{CH}}} ) e^{-\frac{j2\pi k  d }{K}} \end{bmatrix}\right).
%\end{equation}
%\begin{equation}\label{eq:W_hat_k}
%\bD[k] =  \mathrm{diag} \left(\begin{bmatrix}  \alpha_{1}  \omega_{\tau_1}[k] & \cdots &  \alpha_{N_{\rm{CH}}}  \omega_{\tau_{N_{\rm{CH}}} }[k]\end{bmatrix}\right).
%\end{equation}
%\begin{equation}\label{eq:W_hat_k}
%\bD[k] =  \begin{bmatrix}
%\alpha_{1}  \omega_{\tau_1}[k]  & 0 & \cdots & 0 \\ 
%0 &  \alpha_{2}  \omega_{\tau_2}[k] & 0 & \vdots \\
%\vdots & 0 & \ddots & 0 \\
%0 & \cdots & 0 &\alpha_{N_{\rm{CH}}}  \omega_{\tau_{N_{\rm{CH}}} }[k]
%\end{bmatrix}.
%\end{equation}
\begin{equation}\label{eq:W_hat_k}
\bD[k] =  \begin{bmatrix}
\alpha_{1}  \omega_{\tau_1}[k]  & \cdots & 0 \\ 
\vdots &  \ddots & \vdots \\
0 & \cdots & \alpha_{N_{\rm{CH}}}  \omega_{\tau_{N_{\rm{CH}}} }[k]
\end{bmatrix}.
\end{equation}
In the next section, we formulate the hybrid precoding design problem, before presenting our solutions in the following sections.

%%%%%%%%%%%%%%%%%%%%%%%%%%%%%%%%%%%%%%%%%%%%%%%%%%%%%%%%%%%%%%%%%%%%%%%%%%%%%%%%%%%%%%%%%%%%%%%%%%%%%%%%%%%%%
\section{Problem Formulation}\label{sec:Prob}
%%%%%%%%%%%%%%%%%%%%%%%%%%%%%%%%%%%%%%%%%%%%%%%%%%%%%%%%%%%%%%%%%%%%%%%%%%%%%%%%%%%%%%%%%%%%%%%%%%%%%%%%%%%%%
The objective of this paper is to design the hybrid analog and digital precoders at the BS to maximize the mutual information assuming that the transmit symbol at each subcarrier, $s[k]$, has a Gaussian distribution. This problem can be formulated as 
\begin{align}\label{eq:opt_criterion_capacity}
\left\{\Frf^\star, \left\{\Fbb^\star[k]\right\}_{k=1}^K \right\} = & \arg \hspace{-20pt} \max_{ \Frf, \left\{\Fbb[k]\right\}_{k=1}^K}  \sum_{k=1}^{K}  \log \det \left( \bI + \frac{1}{\sigma_{N}^2} \bH[k] \bF_{\rm{RF}} \bF_{\rm{BB}}[k] \bF_{\rm{BB}}^{*}[k] \bF_{\rm{RF}}^{*} \bH^{*}[k] \right) \nonumber \\
& \mathrm{s.t.} \hspace{50pt} \textrm{   } \sum_{k=1}^{K} || \bF_{\rm{RF}} \bF_{\rm{BB}}[k]||_{\rm{F}}^{2} \leq P_{\rm{tot}},
\end{align}
where the precoders must satisfy a total power constraint. One difficulty in solving \eqref{eq:opt_criterion_capacity} is the coupling between the baseband and RF precoders in the power constraint. Using a change of variable trick, though, and setting $\bF_{\rm{BB}}[k] = \left( \bF_{ \rm{RF}}^{*} \bF_{\rm{RF}} \right)^{-\frac{1}{2}} \hat{\bF}_{\rm{BB}}[k]$ with $\hat{\bF}_{\rm{BB}}[k]$ a dummy variable, the problem in \eqref{eq:opt_criterion_capacity} can be equivalently written as \cite{Alkhateeb2015a}
\begin{align}\label{eq:opt_criterion_equiv}
\left\{\Frf^\star, \left\{\hat{\bF}_\mathrm{BB}^\star[k]\right\}_{k=1}^K \right\} = & \arg \hspace{-20pt} \max_{ \Frf, \left\{\hat{\bF}_\mathrm{BB}[k]\right\}_{k=1}^K  }  \sum_{k=1}^{K}  \log \det \left( \bI + \frac{1}{\sigma_{N}^2} \bH_{\rm{eff}}[k]  \hat{\bF}_{\rm{BB}}[k] \hat{\bF}_{\rm{BB}}^{*}[k]  \bH_{\rm{eff}}^{*}[k] \right) \nonumber \\
& \mathrm{s.t.} \hspace{50pt} \textrm{   } \sum_{k=1}^{K} || \hat{\bF}_{\rm{BB}}[k]||_{\rm{F}}^{2} \leq P_{\rm{tot}},
\end{align}
where $\bH_\mathrm{eff}$ is an effective channel matrix defined as
\begin{equation}\label{eq:H_eff}
\bH_{\rm{eff}}[k]  =  \bH[k]  \bF_{\rm{RF}} (\bF^{*}_{\rm{RF}} \bF_{\rm{RF}})^{-\frac{1}{2}}.
\end{equation}

%\begin{equation}\label{eq:pow_constraint}
%\begin{split}
%|| \bF_{\rm{RF}} \bF_{\rm{BB}}[k]||_{\rm{F}}^{2} & =  {\rm{Tr}} \left(   \bF_{\rm{RF}} \bF_{\rm{BB}} [k] \bF_{\rm{BB}}^{*} [k] \bF_{\rm{RF}}^{*} \right) \\
%& =  {\rm{Tr}} \left(   \bF_{\rm{RF}}^{*} \bF_{\rm{RF}} \bF_{\rm{BB}} [k] \bF_{\rm{BB}}^{*} [k]  \right) \\
%& =  {\rm{Tr}} \left(  \left( \bF_{ \rm{RF}}^{*} \bF_{\rm{RF}} \right)^{\frac{1}{2}} \bF_{\rm{BB}} [k] \bF_{\rm{BB}}^{*} [k]  \left( \bF_{ \rm{RF}}^{*} \bF_{\rm{RF}} 
%\right)^{\frac{1}{2}} \right) \\
%&=  {\rm{Tr}} \left( \hat{\bF}_{\rm{BB}} [k] \hat{\bF}_{\rm{BB}}^{*} [k]  \right) \\
%&= || \hat{\bF}_{\rm{BB}}[k]||_{\rm{F}}^{2} ,
%\end{split}
%\end{equation}
%
%where $\hat{\bF}_{\rm{BB}}[k]$ is defined as
%
%\begin{equation}\label{eq:F_BB_hat}
%\hat{\bF}_{\rm{BB}}[k]  = \left( \bF_{ \rm{RF}}^{*} \bF_{\rm{RF}} \right)^{\frac{1}{2}} \bF_{\rm{BB}}[k].
%\end{equation}

If $\bF_{ \rm{RF}}$ is given, and assuming perfect channel knowledge at the transmitter, the digital precoders can be found by using a conventional singular value decomposition (SVD) scheme with respect to the effective channel at each subcarrier. 
%The optimum solution of $\hat{\bF}_{ \rm{BB}}[k]$ is then
%\begin{equation}\label{eq:opt_F_BB_hat}
%\hat{\bF}^\star_{\rm{BB}}[k]=\bV_{\rm{eff}}[k] \bP_{\rm{eff}}^{\frac{1}{2}}[k]
%\end{equation}
%where $\bV_{\rm{eff}}[k]$ is the right singular vectors of the effective channel matrix, $\bH_{\rm{eff}}[k] $, which is decomposed as
%\begin{equation}\label{eq:Heff_SVD}
%\bH_{\rm{eff}}[k]  =  \bU_{\rm{eff}}[k]  \mathbf{\Lambda}_{\rm{eff}}[k]  \bV_{\rm{eff}}^{*}[k], 
%\end{equation}
%and $\bP_{\rm{eff}}[k]$ is a diagonal matrix whose diagonal elements represent the water-filling power control solution with respect to the effective channel singular values. 
Let $\bH_{\rm{eff}}[k] $ be decomposed by SVD as 
\begin{equation}\label{eq:Heff_SVD}
\bH_{\rm{eff}}[k]  =  \bU_{\rm{eff}}[k]  \mathbf{\Lambda}_{\rm{eff}}[k]  \bV_{\rm{eff}}^{*}[k], 
\end{equation}
and let $\bP_{\rm{eff}}[k]$ be a diagonal matrix whose diagonal elements represent the water-filling power control solution with respect to the effective channel singular values. Then, the optimum solution of  $\hat{\bF}_{ \rm{BB}}[k]$ can be represented as
\begin{equation}\label{eq:opt_F_BB_hat}
\hat{\bF}^\star_{\rm{BB}}[k]=\bV_{\rm{eff}}[k] \bP_{\rm{eff}}^{\frac{1}{2}}[k].
\end{equation}
Once the optimal $\hat{\bF}_\mathrm{BB}^\star[k]$ is found, the optimal baseband precoder can be calculated as
\begin{equation}\label{eq:opt_F_BB}
\bF_{\rm{BB}}^\star[k]=(\bF^{*}_{\rm{RF}} \bF_{\rm{RF}})^{-\frac{1}{2}} \hat{\bF}_{\rm{BB}}^\star[k]= (\bF^{*}_{\rm{RF}} \bF_{\rm{RF}})^{-\frac{1}{2}} \bV_{\rm{eff}}[k] \bP_{\rm{eff}}^{\frac{1}{2}}[k].
\end{equation} 

Since the optimal baseband precoding matrices $\bF^\star_{\rm{BB}}[k]$'s depend only on $ \bH[k]$ and  $\bF_{\rm{RF}}$, we can now rewrite the optimization problem in \eqref{eq:opt_criterion_equiv} over $\Frf$ only as
\begin{equation}\label{eq:opt_criterion_equiv_onlyFrf}
\Frf^\star=\arg \max_{ \mathbf{F}_{\rm{RF}}}  \sum_{k=1}^{K}  \sum_{s=1}^{S}   \log \left( 1 + \frac{ \lambda^{2}_{s} \big( \mathbf{H} \left[ k \right] \mathbf{F}_{\rm{RF}} (\bF^{*}_{\rm{RF}} \bF_{\rm{RF}})^{-\frac{1}{2}} \big) p_{s,k} }{\sigma_{N}^2}  \right), 
\end{equation}
where $S=\min(N_{\rm{RF}}, N_{\rm{RX}})$ is the maximum possible number of streams, $\lambda_{s} \left( \bA \right)$ is the $s$-th singular value of $\bA$, and $p_{s,k}$ is the power of the $s$-th stream at the $k$-th subcarrier, which is given by the water-filling power control solution
\begin{equation}\label{eq:WFpower_control}
p_{s,k} = \left( \mu - \frac{\sigma_{N}^2}{  \lambda^{2}_{s} \big( \mathbf{H} \left[ k \right] \mathbf{F}_{\rm{RF}} (\bF^{*}_{\rm{RF}} \bF_{\rm{RF}})^{-\frac{1}{2}} \big)} \right)^{+}, \\
\end{equation}
with $\mu$ satisfying
\begin{equation}\label{eq:WFpower_control_mu}
\sum_{k=1}^{K} \sum_{s=1}^{S} \left( \mu - \frac{\sigma_{N}^2}{  \lambda^{2}_{s} \big( \mathbf{H} \left[ k \right] \mathbf{F}_{\rm{RF}} (\bF^{*}_{\rm{RF}} \bF_{\rm{RF}})^{-\frac{1}{2}} \big)} \right)^{+} = P_{\rm{tot}}. 
\end{equation}
 
Note that the original optimization problem in \eqref{eq:opt_criterion_capacity} is now equivalent to \eqref{eq:opt_criterion_equiv_onlyFrf} where we only need to optimize over $\bF_{\rm{RF}}$. This  problem, though, is non-convex and hard to solve. Therefore, we relax the optimization and instead maximize the sum of the squared singular values of the effective channels. In \sref{subsec:eval_relax}, we will evaluate this relaxation and show that it works well for wideband mmWave channels with practical system and channel parameters. Our relaxed objective is to solve 
\begin{equation}\label{eq:opt_criterion_relaxed}
\Frf^\star=\arg \max_{\mathbf{F}_{\rm{RF}}}  \sum_{k=1}^{K}  \sum_{s=1}^{S} \lambda^{2}_{s} \big( \mathbf{H} [ k ] \mathbf{F}_{\rm{RF}} (\bF^{*}_{\rm{RF}} \bF_{\rm{RF}})^{-\frac{1}{2}}   \big) .
\end{equation}

\noindent Once the optimal RF precoder is found, the water-filling power control is applied with respect to the effective channel singular values associated with $\Frf^\star$. 

It is worth noting here that we have not put any constraints on the implementation of the RF precoders. Typically, the RF precoding is realized using networks of phase shifters with certain hardware limitations, e.g., only constant-modulus and quantized angles may be allowed. These limitations impose additional constraints on the entries of the RF precoding matrix. This will be addressed later in \sref{sec:Fully}, after investigating the more relaxed version in \eqref{eq:opt_criterion_relaxed} with no RF hardware constraints.  

%%%%%%%%%%%%%%%%%%%%%%%%%%%%%%%%%%%%%%%%%%%%%%%%%%%%%%%%%%%%%%%%%%%%%%%%%%%%%%%%%%%%%%%%%%%%%%%%%%%%%%%%%%%%%
\section{Wideband Hybrid Precoding Design for Fully-connected Architectures} \label{sec:Fully}
%%%%%%%%%%%%%%%%%%%%%%%%%%%%%%%%%%%%%%%%%%%%%%%%%%%%%%%%%%%%%%%%%%%%%%%%%%%%%%%%%%%%%%%%%%%%%%%%%%%%%%%%%%%%%
In this section, we consider the system model in \figref{fig:WB_HP_architecture} assuming a fully-connected hybrid architecture where each RF chain is connected to all the $\Nt$ antennas. In the following proposition, we derive the structure of the optimal RF precoders that solve  \eqref{eq:opt_criterion_relaxed}. 
\begin{proposition}
Let $\bR=\frac{1}{K} \sum_{k=1}^{K} \mathbf{H}^{*}[k] \mathbf{H}[k]
$ represent a sample covariance matrix of frequency domain channel vectors, with  eigenvalue decomposition $\bR= \bV_{\rm{R}} \mathbf{\Lambda}_{\rm{R}} \bV^{*}_{\rm{R}}$. Let $[ \bV_{\rm{R}}]_{1:N_{\rm{RF}}}$ denote the matrix with the dominant $N_{\rm{RF}}$ eigenvectors of $\bR$. The solution to \eqref{eq:opt_criterion_relaxed} can then be written as 
\begin{equation}
\bF_{\rm{RF}}^\star = \left[\bV_{\rm{R}}\right]_{1:N_{\rm{RF}}} \bA,
\end{equation}
with an arbitrary $N_\mathrm{RF} \times N_\mathrm{RF}$ full rank matrix $\bA$. 
\label{prop:Fully_connected}
\end{proposition}
\begin{proof}
Let $\bF_{\rm{RF}}$ be decomposed by SVD as $\bF_{\rm{RF}}= \bU_{\rm{RF}} \mathbf{\Lambda}_{\rm{RF}} \bV^{*}_{\rm{RF}}$. Then, we note that the objective function in \eqref{eq:opt_criterion_relaxed} can be also written as 
\begin{align}\label{eq:opt_objective_fun}
\sum_{k=1}^{K}  \sum_{s=1}^{S} \lambda^{2}_{s} \big( \mathbf{H} [k] \mathbf{F}_{\rm{RF}}  (\bF^{*}_{\rm{RF}} \bF_{\rm{RF}})^{-\frac{1}{2}}  \big)  &  \\
& \hspace{-90pt} = \sum_{k=1}^{K} ||\mathbf{H}[k]\mathbf{F}_{\rm{RF}} (\bF^{*}_{\rm{RF}} \bF_{\rm{RF}})^{-\frac{1}{2}} ||^2_{F} \\
& \hspace{-90pt}=  \mathrm{Tr} \left( (\bF^{*}_{\rm{RF}} \bF_{\rm{RF}})^{-\frac{1}{2}} \mathbf{F}_{\rm{RF}} ^*   \left( \sum_{k=1}^{K} \mathbf{H}^{*}[k] \mathbf{H}[k] \right) \mathbf{F}_{\rm{RF}} (\bF^{*}_{\rm{RF}} \bF_{\rm{RF}})^{-\frac{1}{2}}  \right) \\
&\hspace{-90pt}=  \mathrm{Tr} \left( \mathbf{F}_{\rm{RF}}  (\bF^{*}_{\rm{RF}} \bF_{\rm{RF}})^{-1} \mathbf{F}_{\rm{RF}} ^*   \left( \sum_{k=1}^{K} \mathbf{H}^{*}[k] \mathbf{H}[k] \right) \right) \\
&\hspace{-90pt}=  \mathrm{Tr} \left( K  \bU_{\rm{RF}}  \bU^{*}_{\rm{RF}}  \bR   \right) \\
&\hspace{-90pt}= K || \bR^{\frac{1}{2}} \bU_{\rm{RF}}||^2_{F} \label{eq:opt_2}.
\end{align}
Given the eigenvalue decomposition $\bR= \bV_{\rm{R}} \mathbf{\Lambda}_{\rm{R}} \bV^{*}_{\rm{R}}$, the singular vector matrix $\bU_{\rm{RF}}^\star$ that maximizes the objective function in \eqref{eq:opt_criterion_relaxed}, or equivalently \eqref{eq:opt_2}, can now be directly obtained as 
\begin{equation}\label{eq:opt_U_RF}
\bU_{\rm{RF}}^\star= [ \bV_{\rm{R}}]_{1:N_{\rm{RF}}} \bU_{\rm{A}},
\end{equation}
where $\bU_{\rm{A}}$ is an arbitrary $N_{\rm{RF}} \times N_{\rm{RF}}$ unitary matrix that represents the unitary invariance property of the precoding matrix. Given $\bU_\mathrm{RF}^\star$, the optimal RF precoding matrix  $\bF_{\rm{RF}}$ that solves \eqref{eq:opt_criterion_relaxed} can be expressed as
\begin{align}\label{eq:opt_F_RF}
\bF_{\rm{RF}}^\star &=  [ \bV_{\rm{R}}]_{1:N_{\rm{RF}}} \bU_{\rm{A}} \mathbf{\Lambda}_{\rm{RF}} \bV^{*}_{\rm{RF}} \\
&= [ \bV_{\rm{R}}]_{1:N_{\rm{RF}}} \bA,
\end{align}
where $\bA$ is an arbitrary $N_{\rm{RF}} \times N_{\rm{RF}}$ matrix with full rank.
\end{proof}

Next, we show that the solution in Proposition \ref{prop:Fully_connected} achieves the same spectral efficiency as the fully-digital solution to \eqref{eq:opt_criterion_capacity} if the number of RF chains is larger than or equal to the number of channel paths, i.e., $N_\mathrm{RF}\geq N_{\rm{CH}}$. First, we rewrite the sample covariance matrix $\bR$ as
\begin{align}\label{eq:Rmat_using_ch_model} 
\bR & = \frac{1}{K} \sum_{k=1}^{K} \bH^*[k] \bH[k] \\
&= \frac{1}{K}  \sum_{k=1}^{K}  \left( \bA_{{\rm{T}}} \bD^*[k] \bA^*_{{\rm{R}}} \bA_{{\rm{R}}} \bD[k] \bA^*_{{\rm{T}}}  \right) \\
&=\bA_{{\rm{T}}} \left(  \frac{1}{K}  \sum_{k=1}^{K}  \bD^*[k] \bA^*_{{\rm{R}}} \bA_{{\rm{R}}} \bD[k] \right)   \bA^*_{{\rm{T}}} \label{eq:Cov_simple}.
\end{align}
%\noindent From \eqref{eq:Cov_simple}, we note that the maximum rank of the matrix $\bR$ is $\min(N_{\rm{TX}}, N_{\rm{CH}})$, as $\bA_{{\rm{T}}}$ is an $N_{\rm{TX}} \times N_{\rm{CH}}$ matrix, and $\sum_{k=1}^{K}  \bD^*[k] \bA^*_{{\rm{R}}} \bA_{{\rm{R}}} \bD[k] $ is an $N_{\rm{CH}} \times N_{\rm{CH}}$ matrix. 
Note that $\bA_{{\rm{T}}}$ is an $N_{\rm{TX}} \times N_{\rm{CH}}$ matrix, and $\sum_{k=1}^{K}  \bD^*[k] \bA^*_{{\rm{R}}} \bA_{{\rm{R}}} \bD[k] $ is an $N_{\rm{CH}} \times N_{\rm{CH}}$ matrix. As a result, the rank of the matrix $\bR$ is at most $\min(N_{\rm{TX}}, N_{\rm{CH}})$.
As mmWave systems will employ large antenna arrays \cite{Roh2014,HeathJr2015}, and mmWave channels are expected to be sparse \cite{Rappaport2013a,Rappaport2013}, the number of channel paths will likely be less than the number of antennas, i.e., $N_{\rm{CH}} < N_{\rm{TX}}$. In this case, the rank of the channel covariance equals the number of paths, i.e., the matrix $\bR$ becomes rank-deficient. Based on that, the channel matrix at subcarrier $k$ can be represented as
\begin{equation}\label{eq:Hk_using_VR}
\bH[k]= \hat{\bH}[k]  \bV^*_{\rm{R}}, 
\end{equation} 
where $\bV_{\rm{R}} $ is the $N_{\rm{TX}} \times N_{\rm{CH}}$ right singular matrix of $\bR$, and $\hat{\bH}[k]$ is an $N_{\rm{RX}} \times N_{\rm{CH}}$ matrix. Given that, the fully digital precoding solution that solves the optimization problem in \eqref{eq:opt_criterion_capacity} is given by the SVD solution. Let $\bH[k]=\bU[k] \mathbf{\Lambda}[k] \bV^{*}[k]$ define the SVD of the channel matrix $\bH[k]$, then  the fully-digital optimal precoder equals $\bV[k]$, which can be written as
\begin{equation}\label{eq:Vk_opt}
\bV[k] = \bV_{\rm{R}} \hat{\bH}^{*}[k] \bU[k] \mathbf{\Lambda}^{-1}[k].
\end{equation}

\noindent Assuming that the number of RF chains is at least as large as the number of paths, i.e., $N_\mathrm{RF} \geq N_{\rm{CH}}$, then the matrix $\bV[k]$ in \eqref{eq:Vk_opt} can also be rewritten in terms of the derived baseband and RF precoders in Proposition \ref{prop:Fully_connected} as
\begin{equation}\label{eq:Vk_opt2}
\bV[k] = \bF_{\rm{RF}} \bF_{\rm{BB}}[k],
\end{equation}
with $\bF_{\rm{RF}}  = \bV_{\rm{R}} \bA$ and $ \bF_{\rm{BB}}[k] = \bA^{-1} \hat{\bH}^{*}[k] \bU[k] \mathbf{\Lambda}^{-1}[k]$. This means that the derived hybrid precoding solution in Proposition \ref{prop:Fully_connected} represents an optimal solution for \eqref{eq:opt_criterion_capacity}, and achieves the spectral efficiency of the fully-digital architecture when $N_\mathrm{RF} \geq N_{\rm{CH}}$. 

To account for the RF constraints, we approximate the unconstrained RF precoder design in \eqref{eq:opt_F_RF} by the constrained precoder $\hat{\bF}_\mathrm{RF}$ that solves 
\begin{equation} \label{eq:Const_RF}
\hat{\bF}_\mathrm{RF}=\arg\min_{\bX, \left|[\bX]_{m,n}\right|=1}\left\|\bX-\bF_\mathrm{RF}\right\|^2_\mathrm{F},
\end{equation}
which is known to provide a good approximation \cite{ElAyach2014,Alkhateeb2015a}.
\noindent The solution of \eqref{eq:Const_RF} is given by $[\hat{\bF}_\mathrm{RF}]_{m,n}=e^ {j \measuredangle \left( \left[\bF_{\mathrm{RF}}\right]_{m,n} \right)}$, where $\measuredangle( \alpha )$ denotes the phase of a complex number $\alpha$. 
Thanks to the design of the optimal unconstrained RF precoder in \eqref{eq:opt_F_RF}, which depends on the channel singular vectors, and because these singular vectors take a DFT structure for uniform arrays as $N \to \infty$ \cite{Adhikary2013,ElAyach2012a,Alkhateeb2015a}, 
this simple solution can be a reasonable substitute for the unknown optimal solution which needs further study. This will be shown by numerical simulations in \sref{sec:Results}.
%the proposed approximation for the constrained RF precoder approaches the unconstrained solution as the number of antennas increases. This will be shown by numerical simulations in \sref{sec:Results}.

%%%%%%%%%%%%%%%%%%%%%%%%%%%%%%%%%%%%%%%%%%%%%%%%%%%%%%%%%%%%%%%%%%%%%%%%%%%%%%%%%%%%%%%%%%%%%%%%%%%%%%%%%%%%%
\section{Wideband Hybrid Precoding Design for Fixed Subarray Architectures} \label{sec:SubArray_Fixed}
%%%%%%%%%%%%%%%%%%%%%%%%%%%%%%%%%%%%%%%%%%%%%%%%%%%%%%%%%%%%%%%%%%%%%%%%%%%%%%%%%%%%%%%%%%%%%%%%%%%%%%%%%%%%%
In this section, we consider the hybrid architecture in \sref{sec:Model}, but assuming a subarray structure \cite{ElAyach2014,Han2015}. This means that every RF chain is connected to only a subset of the antennas with $N_\mathrm{sub}=\frac{\Nt}{N_\mathrm{RF}}$ elements. We assume that $ \Nt$ is a multiple of $N_\mathrm{RF}$. Let the antenna indexes be $\{ 1,\cdots,N_\mathrm{TX} \}$ and  $\mathcal{S}_r$ denote the partitioned subset of antenna indexes connected to the $r$-th RF chain such as
\begin{equation}\label{eq:Fixed_subset_ex}
\begin{matrix}
\mathcal{S}_1=\{1,\cdots,N_\mathrm{sub} \} \\
\mathcal{S}_2=\{N_\mathrm{sub} +1,\cdots,2N_\mathrm{sub} \} \\
 \vdots \\
 \mathcal{S}_{N_\mathrm{RF}}=\{(N_\mathrm{RF}-1)N_\mathrm{sub}+1,\cdots,N_\mathrm{RF} N_\mathrm{sub} \}.
\end{matrix}
\end{equation}
With this architecture, the analog RF precoding matrix, $\bF_{\rm{RF}} $, has the form of a block diagonal matrix as
%\begin{equation}\label{eq:F_RF_subarray}
%\bF_{\rm{RF}} = \begin{bmatrix}
%\bff_{{\rm{RF}},\mathcal{S}_1} & 0 & \cdots & 0 \\ 
%0 &  \bff_{{\rm{RF}},\mathcal{S}_2} & 0 & \vdots \\
%\vdots & 0 & \ddots & 0 \\
%0 & \cdots & 0 & \bff_{{\rm{RF}},\mathcal{S}_{N_{\rm{RF}}}}
%\end{bmatrix},
%\end{equation}
\begin{equation}\label{eq:F_RF_subarray}
\bF_{\rm{RF}} = \begin{bmatrix}
\bff_{{\rm{RF}},\mathcal{S}_1} &  \cdots & \mathbf{0} \\ 
\vdots & \ddots & \vdots \\
\mathbf{0} & \cdots &  \bff_{{\rm{RF}},\mathcal{S}_{N_{\rm{RF}}}}
\end{bmatrix},
\end{equation}
where $\bff_{{\rm{RF}},\mathcal{S}_r}$ is an $N_{{\rm{sub}}} \times 1$ analog beamforming vector associated with the $r$-th RF chain. This is a distinct property compared to the fully-connected case whose analog precoding matrix takes the form 
%$\bF_{\rm{RF}} = \left[\bff_{{\rm{RF}},1}, \bff_{{\rm{RF}},2}, ..., \bff_{{\rm{RF}},N_{\rm{RF}}}\right]$, 
$\bF_{\rm{RF}} = \begin{bmatrix} \bff_{{\rm{RF}},1} & \bff_{{\rm{RF}},2}& ...& \bff_{{\rm{RF}},N_{\rm{RF}}}\end{bmatrix}$, 
with  $\bff_{{\rm{RF}},r}$ an $N_{{\rm{TX}}} \times 1$ analog beamforming vector associated with the $r$-th RF chain. Given this subarray architecture, the overall $N_{{\rm{RX}}} \times N_{{\rm{TX}}}$ channel matrix can be expressed using each subarray channel matrix as
\begin{equation}\label{eq:Hk_subarray}
\bH[k] = \begin{bmatrix}
\bH_{\mathcal{S}_1}[k] &  \bH_{\mathcal{S}_2}[k] & \cdots & \bH_{\mathcal{S}_{N_{\rm{RF}}}}[k]
\end{bmatrix},
\end{equation}
where $\bH_{\mathcal{S}_r}[k] $ is the $N_{{\rm{RX}}} \times N_{{\rm{sub}}}$ channel matrix of the  $r-$th subarray. 
Next, we present Proposition \ref{prop:Sub_Arrays} that obtains the structure of the optimal hybrid precoders solving \eqref{eq:opt_criterion_relaxed} under the subarray architecture. 
\begin{proposition} \label{prop:Sub_Arrays}
The $\Nt \times N_\mathrm{RF}$ RF precoder $\bF_\mathrm{RF}$ that solves \eqref{eq:opt_criterion_relaxed} under the subarray hybrid analog/digital architecture is given by $\bF^\star_\mathrm{RF}=\mathrm{blkdiag}\left(\bff^\star_{\mathrm{RF},\mathcal{S}_1}, ..., \bff^\star_{\mathrm{RF}, \mathcal{S}_{N_\mathrm{RF}}}\right)$, with
\begin{equation}
\bff^\star_{{\rm{RF}},\mathcal{S}_r}=\alpha_{r} \mathbf{v}_{\bR_{\mathcal{S}_r},1} , \;\; {\rm{for}} \;\; r={1,\cdots,N_{\rm{RF}}},
\end{equation} 
where $\alpha_{r}$ is an arbitrary complex value, and $\mathbf{v}_{\bR_{\mathcal{S}_r},1} $ is the largest singular vector of the covariance matrix $\bR_{\mathcal{S}_r}$, which is associated with the $r$-th subarray channel matrix and is defined as
\begin{equation} \label{eq:Cov_Subarray}
\bR_{\mathcal{S}_r}=\frac{1}{K} \sum_{k=1}^{K} \mathbf{H}^{*}_{\mathcal{S}_r}[k] \mathbf{H}_{\mathcal{S}_r}[k], \;\; {\rm{for}} \;\;r={1,\cdots,N_{\rm{RF}}}.
\end{equation}
\end{proposition}
\begin{proof} 
%For the array-of-subarray architecture, the objective function of the optimization problem in \eqref{eq:opt_criterion_relaxed} can be written as
From \eqref{eq:F_RF_subarray}, $\left(\bF^{*}_{\rm{RF}} \bF_{\rm{RF}}\right)^{-\frac{1}{2}}$ has a form of a diagonal matrix as 
%\begin{equation}\label{eq:F_RF_hat_F_RF}
%(\bF^{*}_{\rm{RF}} \bF_{\rm{RF}})^{-\frac{1}{2}}  = \begin{bmatrix}
%|\bff_{{\rm{RF}},\mathcal{S}_1}|^{-1} & 0 & \cdots & 0 \\ 
%0 &  |\bff_{{\rm{RF}},\mathcal{S}_2}|^{-1} & 0 & \vdots \\
%\vdots & 0 & \ddots & 0 \\
%0 & \cdots & 0 & |\bff_{{\rm{RF}},\mathcal{S}_{N_{\rm{RF}}}}|^{-1}
%\end{bmatrix}.
%\end{equation}
\begin{equation}\label{eq:F_RF_hat_F_RF}
(\bF^{*}_{\rm{RF}} \bF_{\rm{RF}})^{-\frac{1}{2}}  = \begin{bmatrix}
|\bff_{{\rm{RF}},\mathcal{S}_1}|^{-1} & \cdots & 0 \\ 
\vdots & \ddots & \vdots \\
0 & \cdots &  |\bff_{{\rm{RF}},\mathcal{S}_{N_{\rm{RF}}}}|^{-1}
\end{bmatrix}.
\end{equation}

\noindent This property of $\left(\bF^{*}_{\rm{RF}} \bF_{\rm{RF}}\right)^{-\frac{1}{2}}$ implies that the effective channel for the $k-$th subcarrier $\bH_{\rm{eff}} [k] $ in \eqref{eq:H_eff} can be written as  
\begin{equation}\label{eq:F_eff_sub}
\begin{split} 
\bH_{\rm{eff}} [k] &= \mathbf{H} \left[ k \right] \mathbf{F}_{\rm{RF}}  \left( \bF^{*}_{\rm{RF}} \bF_{\rm{RF}} \right)^{-\frac{1}{2}} \\
&= \begin{bmatrix}
\frac{\bH_{\mathcal{S}_1}[k] \bff_{{\rm{RF}},\mathcal{S}_1}}{|\bff_{{\rm{RF}},\mathcal{S}_1}|} &  \frac{\bH_{\mathcal{S}_2}[k] \bff_{{\rm{RF}},\mathcal{S}_2}}{|\bff_{{\rm{RF}},\mathcal{S}_2}|} & \cdots & \frac{\bH_{\mathcal{S}_{N_{\rm{RF}}}}[k] \bff_{{\rm{RF}},\mathcal{S}_{N_{\rm{RF}}}}}{|\bff_{{\rm{RF}},\mathcal{S}_{N_{\rm{RF}}}}|} \end{bmatrix}.
\end{split} 
\end{equation}
From \eqref{eq:Cov_Subarray} and \eqref{eq:F_eff_sub}, the objective function of the optimization problem in \eqref{eq:opt_criterion_relaxed} can be written as
\begin{align}\label{eq:objective_fun_sub}
\sum_{k=1}^{K}  \sum_{s=1}^{S} \lambda^{2}_{s} \big( \mathbf{H} \left[ k \right] \mathbf{F}_{\rm{RF}}  (\bF^{*}_{\rm{RF}} \bF_{\rm{RF}})^{-\frac{1}{2}}  \big) &= \sum_{k=1}^{K} ||\bH_{\rm{eff}} [k]  ||^2_{F} \\
&=\sum_{k=1}^{K}  \sum_{r=1}^{N_{\rm{RF}}} \frac{|\bH_{\mathcal{S}_r}[k] \bff_{{\rm{RF}},\mathcal{S}_r}|^2}{|\bff_{{\rm{RF}},\mathcal{S}_r}|^2}  \\
%&=  \sum_{r=1}^{N_{\rm{RF}}} \frac{ \bff_{{\rm{RF}},r}^* \left(\sum_{k=1}^{K} \bH_{r}[k]^* \bH_{r}[k] \right) \bff_{{\rm{RF}},r}}{|\bff_{{\rm{RF}},r}|^2}  \\
&=  \sum_{r=1}^{N_{\rm{RF}}} \frac{ K \bff_{{\rm{RF}},\mathcal{S}_r}^* \bR_{\mathcal{S}_r} \bff_{{\rm{RF}},\mathcal{S}_r}}{|\bff_{{\rm{RF}},\mathcal{S}_r}|^2},
\end{align}
where the third equality comes from \eqref{eq:Cov_Subarray}. 
The maximum value of the objective function in \eqref{eq:objective_fun_sub} can then be written as
\begin{align}\label{eq:max_val_sub}
\max_{\bF_{\rm{RF}}} \sum_{k=1}^{K}  \sum_{s=1}^{S} \lambda^{2}_{s} \big( \mathbf{H} \left[ k \right] \mathbf{F}_{\rm{RF}}  (\bF^{*}_{\rm{RF}} \bF_{\rm{RF}})^{-\frac{1}{2}}  \big) 
&=  \max_{\bff_{\rm{RF},\mathcal{S}_1},\dots,\bff_{{\rm{RF}},\mathcal{S}_{N_{\rm{RF}}}}}  \sum_{r=1}^{N_{\rm{RF}}} \frac{ K \bff_{{\rm{RF}},\mathcal{S}_r}^* \bR_{\mathcal{S}_r} \bff_{{\rm{RF}},\mathcal{S}_r}}{|\bff_{{\rm{RF}},\mathcal{S}_r}|^2} \\
&=  K \sum_{r=1}^{N_{\rm{RF}}} \lambda_{1} \left(  \bR_{\mathcal{S}_r} \right) ,\label{eq:Cov_Sing}
\end{align} 
where $\lambda_{1} \left(  \bA \right) $ denotes the largest singular value of a matrix $\bA $. This maximum value is achieved when the analog beamforming vector for each RF chain $r$ has the structure
\begin{equation}\label{eq:f_rf_opt_sub}
\bff_{{\rm{RF}},\mathcal{S}_r}^\star=\alpha_{r} \mathbf{v}_{\bR_{\mathcal{S}_r},{1}} , \;\; {\rm{for}} \;\; r={1,\cdots,N_{\rm{RF}}},
\end{equation}
where $\alpha_{r}$ is an arbitrary complex value, and $\mathbf{v}_{\bR_{r},{1}} $ is the largest singular vector of $\bR_{r}$. 
\end{proof}

Note that the maximum value of the objective function in \eqref{eq:Cov_Sing} is the sum of the largest singular values of $N_{\rm{RF}}$ submatrices, $\bR_{\mathcal{S}_1},...,\bR_{\mathcal{S}_{N_{\rm{RF}}}}$. This is a distinguishing feature from the fully-connected case where the maximum value is the sum of largest $N_{\rm{RF}}$ singular values of the total matrix, $\bR$, as
\begin{equation}\label{eq:max_obj_val_fully}
\max_{\bF_{\rm{RF}}} \sum_{k=1}^{K}  \sum_{s=1}^{S} \lambda^{2}_{s} \big( \mathbf{H} \left[ k \right] \mathbf{F}_{\rm{RF}}  (\bF^{*}_{\rm{RF}} \bF_{\rm{RF}})^{-\frac{1}{2}}  \big) 
=  K \sum_{r=1}^{N_{\rm{RF}}} \lambda_{r} \left(  \bR \right).
\end{equation} 
While the value of \eqref{eq:max_obj_val_fully} is constant if $\bR$ is given, the value of \eqref{eq:Cov_Sing} depends on the configuration of the submatrices,   $\bR_{\mathcal{S}_1},...,\bR_{\mathcal{S}_{N_{\rm{RF}}}}$. This motivates a dynamic subarray technique, which will be explained in the next section.
%%%%%%%%%%%%%%%%%%%%%%%%%%%%%%%%%%%%%%%%%%%%%%%%%%%%%%%%%%%%%%%%%%%%%%%%%%%%%%%%%%%%%%%%%%%%%%%%%%%%%%%%%%%%%
\section{Wideband Hybrid Precoding Design for Dynamic Subarray Architectures} \label{sec:SubArray_Dynamic}
%%%%%%%%%%%%%%%%%%%%%%%%%%%%%%%%%%%%%%%%%%%%%%%%%%%%%%%%%%%%%%%%%%%%%%%%%%%%%%%%%%%%%%%%%%%%%%%%%%%%%%%%%%%%%

\begin{figure}[t]
	\centering
	\subfigure[center][{Fully-connected structure}]{
		\includegraphics[width=.31\columnwidth]{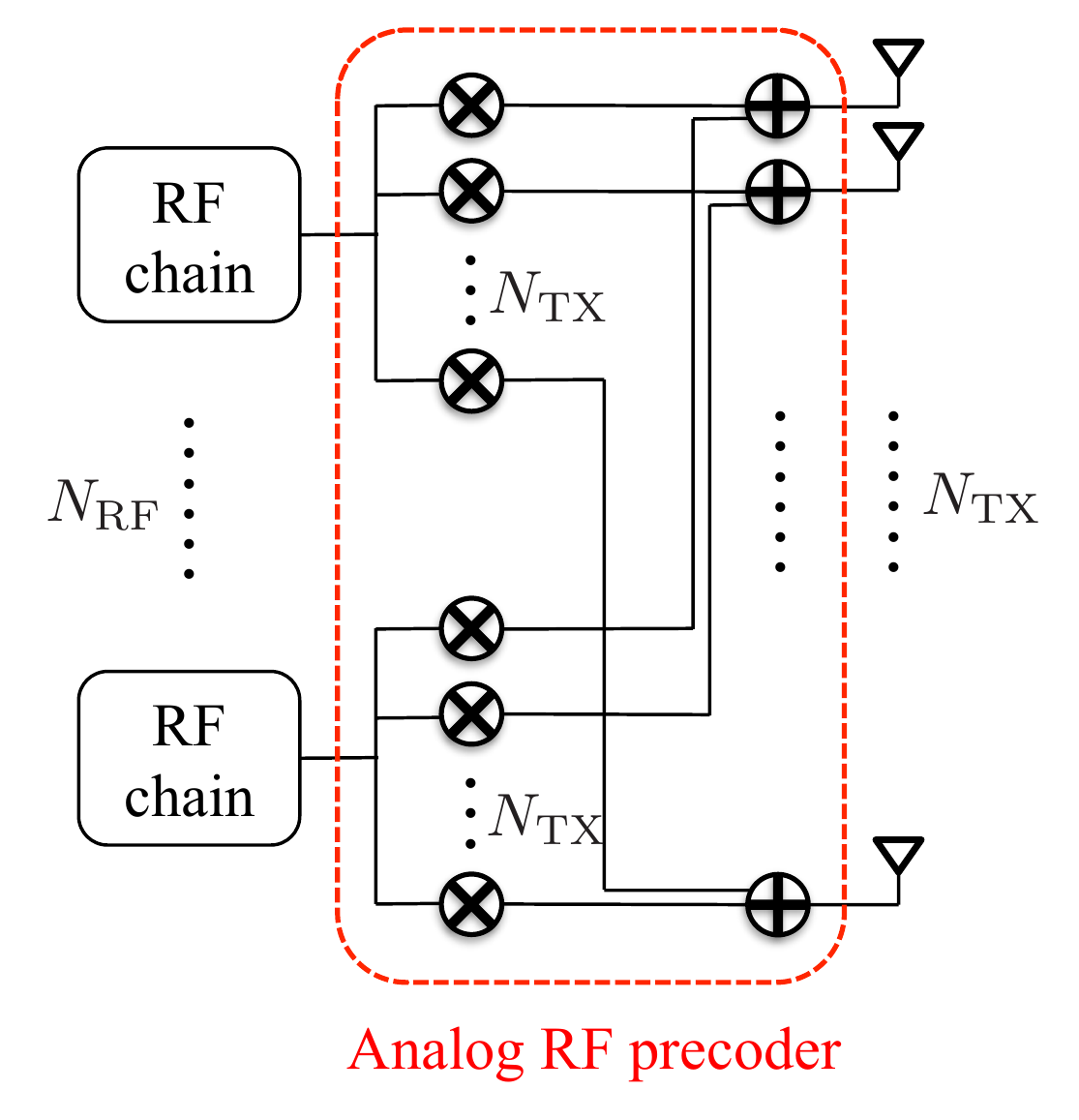}
		\label{fig:Fig_Hybrid_structure_fully_connected}}
	\subfigure[center][{Subarray structure (fixed)}]{
		\includegraphics[width=.31\columnwidth]{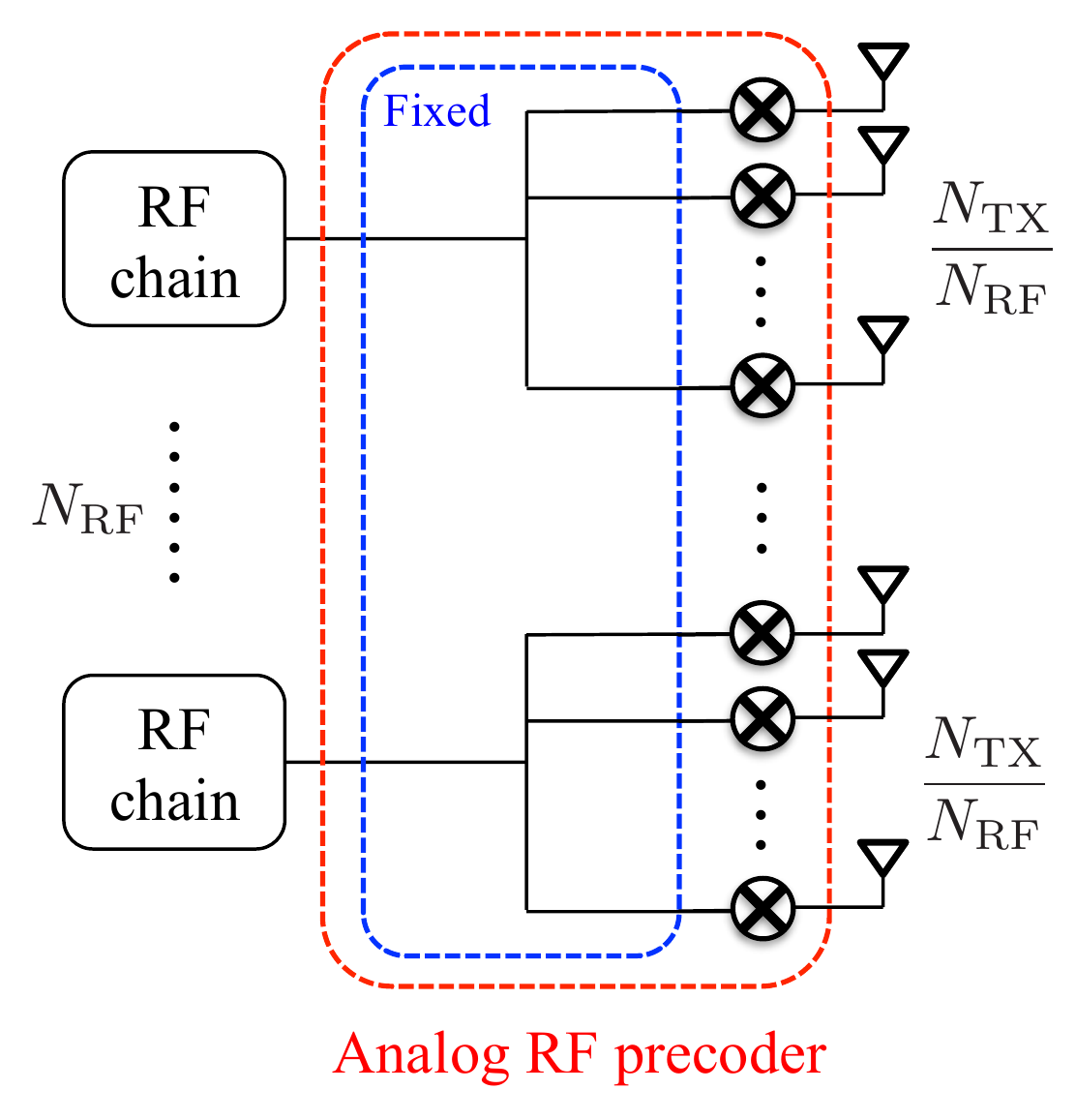}
		\label{fig:Fig_Hybrid_structure_subarray_fixed2}}
	\subfigure[center][{Subarray structure (dynamic)}]{
		\includegraphics[width=.31\columnwidth]{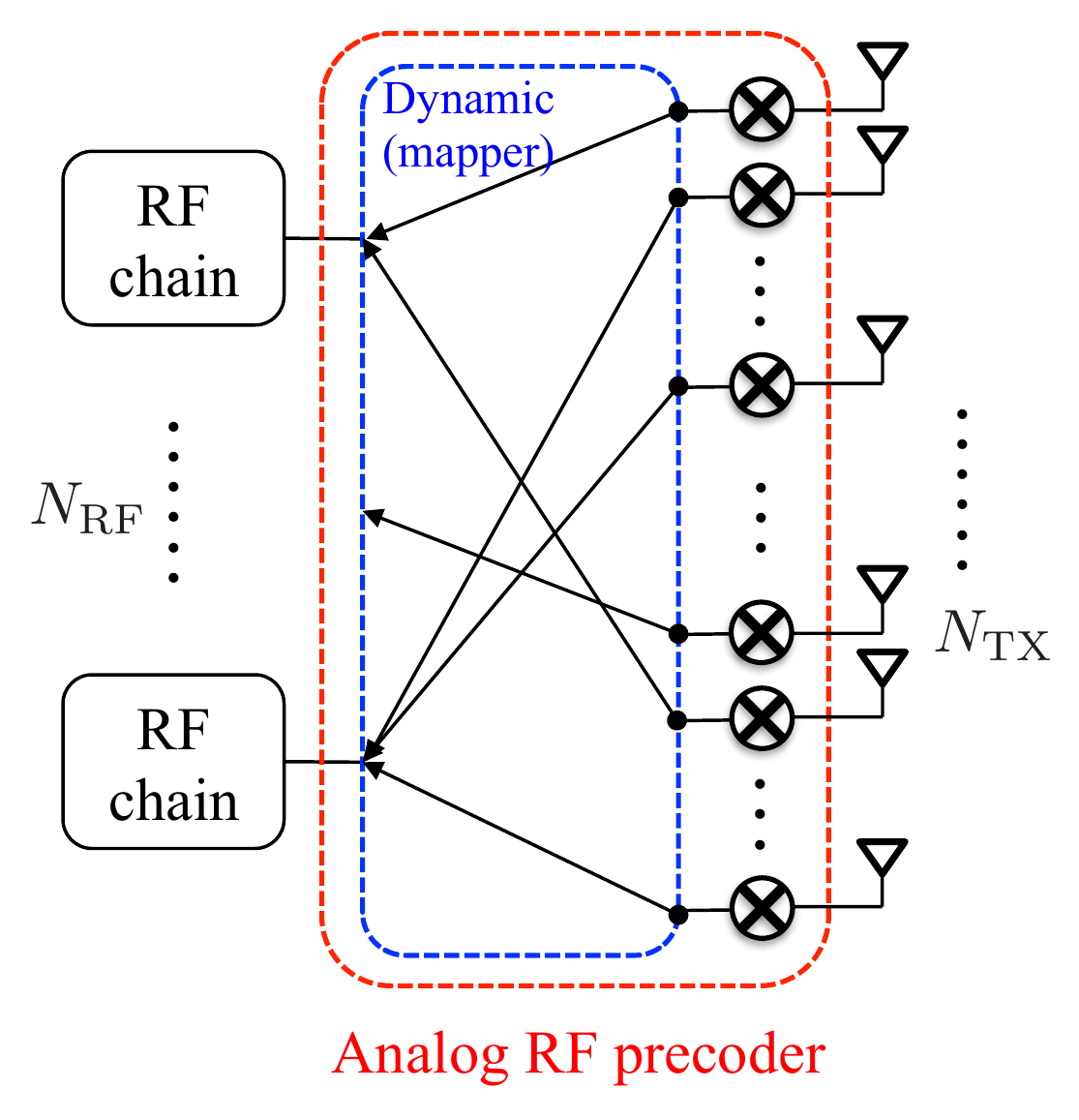}
		\label{fig:Fig_Hybrid_structure_subarray_dynamic}}
	\caption{Hybrid precoding structure with regard to the analog RF precoder type}
	\label{fig:Fig_Hybrid_precoding_structure}
\end{figure}

The subarray hybrid precoding architecture adopted in \sref{sec:SubArray_Fixed} is the conventional one discussed in prior work \cite{Han2015,ElAyach2013}, where each RF chain is connected to a fixed set of adjacent antenna elements. \sref{sec:SubArray_Fixed} shows that the optimal value of the relaxed mutual information objective function in \eqref{eq:max_val_sub} depends on the sum of the largest singular values of the sample covariance matrices associated to these subarrays. As these largest singular values rely on the selected antennas in each of these subsets, then the question that arises is how much gain can be obtained if these subarrays are dynamically adapted to the long-term channel conditions? Implementing the switch matrix required for the dynamic subarrays shown in \figref{fig:Fig_Hybrid_structure_subarray_dynamic} is a topic for future work. The objective of this section is to explore the potential gains in terms of the system spectral efficiency.

Now, we define the dynamic subarray problem. We want an algorithm to partition a set of $N_{\rm{TX}}$ antennas into $N_{\rm{RF}}$ non-empty subsets to maximize the sum of the largest singular values of the submatrices in \eqref{eq:Cov_Sing}. 
%Let $\mathcal{S}_r$ denote the partitioned subset of antennas connected to the $r$-th RF chain. 
Note that each antenna should be included only once in one of these subsets, and the union of all these subsets should be the total set of all antenna indexes, $\{1,...,N_\mathrm{TX}\}$. Contrary to the fixed subarray architectures, each subset $\mathcal{S}_{r}$ can have different cardinalities in the dynamic subarray structure. Then, this dynamic subarray partitioning problem to maximize the objective function in \eqref{eq:opt_criterion_relaxed} can be formulated as
\begin{align}\label{eq:opt_criterion_dynamic_subarray}
\left\{\mathcal{S}^\star_r\right\}_{r=1}^{N_\mathrm{RF}}=&\arg \max_{ \mathcal{S}_{1},\dots,\mathcal{S}_{N_{\rm{RF}}}}  \sum_{r=1}^{N_{\rm{RF}}} \lambda_{1}\left(\bR_{\mathcal{S}_{r}} \right) \\
& \textrm{s.t.} \;\; \bigcup_{r=1}^{N_{\rm{RF}}}\mathcal{S}_{r}=  \left\{1, \cdots,N_{{\rm{TX}}} \right\}, \; \mathcal{S}_{i} \cap \mathcal{S}_{j} =\emptyset \; \textrm{for} \; i \neq j, \;\; \left\vert{\mathcal{S}_{r}}\right\vert > 0 \;\; \forall r \nonumber .
\end{align}

The problem in \eqref{eq:opt_criterion_dynamic_subarray} is a combinatorial optimization problem for which finding the optimal solution requires an exhaustive search for all possible cases. The total number of combinations is known as Stirling number of the second kind and is given by
\begin{equation}\label{eq:number_of_cases_Striling_number}
\frac{1}{(N_{\rm{RF}})!}\sum_{k=0}^{N_{\rm{RF}}} (-1)^{N_{\rm{RF}}-k} \binom{N_{\rm{RF}}}{k} k^{N_{\rm{TX}}},
\end{equation}
which is a large number even for a small number of antennas and RF chains. For example, this number becomes $1.7 \times 10^{8}$ even for 16 transmit antennas and 4 RF chains. One possible suboptimal solution is to assume that all the subsets have the same size, $|\mathcal{S}_r| = N_{\rm{TX}}/N_{\rm{RF}}, \forall r$. Even in this case, though, the total number of combinations is given by $\frac{ \left(N_{\rm{TX}} \right)! }{ \left(\left(\frac{N_{\rm{TX}}}{N_{\rm{RF}}}\right)!\right)^{N_{\rm{RF}}} \left(N_{\rm{RF}}\right)! } $
%\begin{equation}\label{eq:number_of_cases_subOpt}
%\frac{ \left(N_{\rm{TX}} \right)! }{ \left(\left(\frac{N_{\rm{TX}}}{N_{\rm{RF}}}\right)!\right)^{N_{\rm{RF}}} \left(N_{\rm{RF}}\right)! } ,
%\end{equation}
, which is still large, e.g., $2.6 \times 10^{6}$ even for 16 transmit antennas and 4 RF chains.

The objective of this section is to develop a low-complexity yet reasonable solution to the problem in \eqref{eq:opt_criterion_dynamic_subarray}. First, we note that in many cases, calculating  the largest singular values, which is required in \eqref{eq:opt_criterion_dynamic_subarray}, does not have a closed form expression and must be calculated numerically, e.g. through an iterative algorithm \cite{Wolkowicz1980}. Having a closed-form expression of the largest singular value in \eqref{eq:opt_criterion_dynamic_subarray} is important for our subarrays selection problem. To address this challenge, we propose to use a normalized Minkowski $\ell_1$-norm \cite{Lutkepohl2007}, which gives a good approximation of the largest singular value as will be discussed in Proposition \ref{prop:bounds}. Given the overall channel covariance matrix $\bR$, the approximate largest singular value of the subset $\mathcal{S}$ is defined as 
\begin{equation}\label{eq:approx_singular_value}
\hat{\lambda}_1 \left( \bR_{\mathcal{S}} \right) \triangleq \frac{1}{|\mathcal{S}|}\sum_{i=1}^{|\mathcal{S}|}\sum_{j=1}^{| \mathcal{S}|} |[\bR_{\mathcal{S}}]_{i,j} | = \frac{1}{|\mathcal{S}|}\sum_{i \in \mathcal{S}}\sum_{j \in \mathcal{S}} | [\bR]_{i,j} | ,
\end{equation}
where $\sum_{i=1}^{|\mathcal{S}|}\sum_{j=1}^{| \mathcal{S}|} |[\bR_{\mathcal{S}}]_{i,j} |$ is known as the Minkowski $\ell_1$-norm of the matrix $\bR_{\mathcal{S}}$\cite{Lutkepohl2007}. 

This approximate value has two useful properties. First, this value lies between the existing lower and upper bounds on the exact value of the largest singular value as will be proved shortly in Proposition \ref{prop:bounds}. These lower and upper bounds on the largest singular value of $\bR_{\mathcal{S}}$ (with real eigenvalues) are given by \cite{Wolkowicz1980}
\begin{equation}\label{eq:upper_and_lower_bound}
 \lambda_{1, \rm{LB}} \left( \bR_{\mathcal{S}} \right) \leq \lambda_{1} \left( \bR_{\mathcal{S}} \right) \leq  \lambda_{1, \rm{UB}} \left( \bR_{\mathcal{S}} \right),
 \end{equation}
with the lower and upper bounds
\begin{equation}
\begin{split}
 \lambda_{1, \rm{LB}} \left( \bR_{\mathcal{S}} \right)&= m+\frac{s}{(|\mathcal{S}|-1)^{\frac{1}{2}} }\\
 \lambda_{1, \rm{UB}} \left( \bR_{\mathcal{S}} \right)&= m+s(|\mathcal{S}|-1)^{\frac{1}{2}},
\end{split}\label{eq:upper_and_lower_bound_def}
\end{equation}
where
\begin{equation}\label{eq:m_and_s}
m=\frac{{\rm{Tr}}(\bR_{\mathcal{S}})}{|\mathcal{S}|}, \;\;\; s=\left( \frac{{\rm{Tr}}(\bR_{\mathcal{S}}^2)}{|\mathcal{S}|} - m^2 \right) ^{\frac{1}{2}}.
\end{equation}
\noindent In the next proposition, we prove that the approximate largest singular value in \eqref{eq:approx_singular_value} also lies between the existing lower and upper bounds in \eqref{eq:upper_and_lower_bound_def}. 
\begin{proposition}
The approximate value of the largest singular value in \eqref{eq:approx_singular_value} has the same lower and upper bound as those of the exact value of the largest singular value if the matrix is Hermitian with identical diagonal elements.   
\begin{equation}\label{eq:lower_and_upper_bound_approx}
 \lambda_{1, \rm{LB}} \left( \bR_{\mathcal{S}} \right) \leq \hat{\lambda}_{1} \left( \bR_{\mathcal{S}} \right) \leq  \lambda_{1, \rm{UB}} \left( \bR_{\mathcal{S}} \right).
\end{equation}
\label{prop:bounds}
\end{proposition}

\begin{proof}[Proof] See Appendix \ref{app:prop_bounds}.

\end{proof}

Note that the channel covariance matrix is a Hermitian matrix and its diagonal elements tend to be identical if all the antennas are located in the same base station because the path loss term is common to all antennas.   

The second property of the approximate largest singular value is that this approximate value is a tight lower bound of the exact value in the exponential correlation model case.
%comes when the well-known exponential correlation model is used. 
Even though this correlation model cannot perfectly describe the characteristics of $\bR$ matrix in the geometric channel model, this can provide an insight to how close the approximate value in \eqref{eq:approx_singular_value} is to the exact value, due to its analytical tractability. The spatial channel covariance matrix in the exponential correlation model is 
%\begin{equation}\label{eq:Rmat_exponential_corr}
%\bR_{\mathcal{S}} = \begin{bmatrix}
%1 & \rho & \rho^2 & \cdots & \rho^{n-1} \\ 
%\rho^{*} & 1&  \rho & \cdots & \rho^{n-2}  \\
%(\rho^{*})^2 & \rho^{*} & 1&  \cdots & \vdots \\
%\vdots & \vdots & \vdots &  \ddots & \rho \\
%(\rho^{*})^{n-1} & (\rho^{*})^{n-2}& \cdots & \rho^{*} & 1 
%\end{bmatrix},
%\end{equation}
\begin{equation}\label{eq:Rmat_exponential_corr}
\bR_{\mathcal{S}} = \begin{bmatrix}
1 & \rho &  \cdots & \rho^{n-1} \\ 
\rho^{*} & 1& \cdots & \rho^{n-2}  \\
\vdots  & \vdots &  \ddots & \vdots \\
(\rho^{*})^{n-1} &(\rho^{*})^{n-2} & \cdots & 1 
\end{bmatrix},
\end{equation}
where $\rho$ is a complex value whose amplitude is less than or equal to 1. The tight lower bound of the largest singular value in this exponential correlation model is known as \cite{Choi2014} 
\begin{equation}\label{eq:lower_bound_exponential_corr}
\lambda_{1} \left( \bR_{\mathcal{S}} \right) \geq \lambda_{1,\rm{LB(exp)}} \left( \bR_{\mathcal{S}} \right),  
\end{equation}
where
\begin{equation}\label{eq:lower_bound_def_exponential_corr}
\lambda_{1,\rm{LB(exp)}} \left( \bR_{\mathcal{S}} \right) = \frac{1+|\rho|}{1-|\rho|}-\frac{2|\rho| \left( 1- |\rho|^{|\mathcal{S}|}  \right) }{|\mathcal{S}| \left( 1-|\rho| \right)^2} .
\end{equation}

In the next proposition, we show that the approximate largest singular value in \eqref{eq:approx_singular_value} can be regarded as a tight lower bound of the exact largest singular value in the exponential correlation model case.
\begin{proposition}
The approximate value of the largest singular value in \eqref{eq:approx_singular_value} is the same as the tight lower bound of the exact value of the largest singular value if the matrix is modeled as the exponential correlation matrix.
\begin{equation}\label{eq:approx_exponential_corr}
\hat{\lambda}_{1} \left( \bR_{\mathcal{S}} \right)  = \lambda_{1,\rm{LB(exp)}} \left( \bR_{\mathcal{S}} \right) . 
\end{equation}
\end{proposition}

\begin{proof}[Proof]
When $\bR_{\mathcal{S}}$ is modeled as an exponential correlation model as in  \eqref{eq:Rmat_exponential_corr}, the approximate value in \eqref{eq:approx_singular_value} can be calculated as
\begin{equation}\label{eq:approx_singular_value_exp}
\begin{split}
\hat{\lambda}_{1} \left( \bR_{\mathcal{S}} \right) &= \frac{1}{ |\mathcal{S}|}\sum_{i=1}^{ |\mathcal{S}|}\sum_{j=1}^{ |\mathcal{S}|} |[\bR_{ \mathcal{S}}]_{i,j} | \\
&=  \frac{1}{|\mathcal{S}|} \left( |\mathcal{S}|+2\sum^{|\mathcal{S}|-1}_{i=1}\sum^{i}_{j=1} |\rho|^{j} \right) \\
&= \frac{1+|\rho|}{1-|\rho|}-\frac{2|\rho| \left( 1- |\rho|^{|\mathcal{S}|} \right) }{|\mathcal{S}| \left( 1-|\rho| \right)^2} ,
\end{split}
\end{equation}
which is equal to $\lambda_{1,\rm{LB(exp)}} \left( \bR_{\mathcal{S}} \right) $.
\end{proof}

\begin{algorithm}\label{alg:proposed}
\caption{Dynamic subarray partitioning}
\begin{algorithmic}
\State Input:  $\bR$, $N_{\rm{RF}}$, $N_{\rm{TX}}$
\State $\mathcal{S}_{0}=\{1,\dots,N_{\rm{RF}} \}$,   $n_{\rm{sel}}=0$
\State Sort $|[\bR]_{i,j}| $ for $1 \leq i<j \leq N_{\rm{TX}} $ in descending order
\State $\left( |[\bR]_{i_{1},j_{1}}| \geq \cdots \geq |[\bR]_{i_{k},j_{k}}|  \geq \cdots \geq |[\bR]_{i_{K},j_{K}}|, \;\; K=\frac{N_{\rm{TX}}(N_{\rm{TX}}-1)}{2}, \;\;  1 \leq i_k < j_k \leq N_{\rm{TX}} \right)$

\For {$k =1:K $}

\If {$i_{k}, j_{k} \in \mathcal{S}_{0}$}
\If {$n_{\rm{sel}} < N_{\rm{RF}}$}
\State $n_{\rm{sel}}  \gets n_{\rm{sel}} +1 $, $\mathcal{S}_{n_{\rm{sel}}}  \gets \{ i_{k}, j_{k} \}$, $ \mathcal{S}_{0}  \gets \mathcal{S}_{0} \setminus \{ i_{k}, j_{k} \}$
\Else
\State $\hat{r}=\arg\max_{r \in \{1,\dots,N_{\rm{RF}} \}} \left( f_{\bR,N_{\rm{RF}}} \left( \mathcal{S}_{r}  \cup \{ i_{k}, j_{k} \} ,n_{\rm{sel}}, r \right) - f_{\bR,N_{\rm{RF}}} \left( \mathcal{S}_{r} ,n_{\rm{sel}},r \right) \right)$
\State  $\mathcal{S}_{\hat{r}} \gets \mathcal{S}_{\hat{r}}  \cup \{ i_{k}, j_{k} \} $,  $\mathcal{S}_{0} \gets \mathcal{S}_{0} \setminus \{ i_{k}, j_{k} \} $
\EndIf

\ElsIf {$i_{k} \in \mathcal{S}_{m}, \;\; j_{k} \in \mathcal{S}_{l}$ for some $m,l \in \{0,1,\dots,n_{\rm{sel}} \}$ and $m \neq l$}
\State $\mu_{\rm{current}} =  f_{\bR,N_{\rm{RF}}} \left( \mathcal{S}_{m}  ,n_{\rm{sel}},m \right) +  f_{\bR,N_{\rm{RF}}} \left( \mathcal{S}_{l} ,n_{\rm{sel}},l  \right) $
\State $\mu_{{\rm{new}},j} =  f_{\bR,N_{\rm{RF}}} \left( \mathcal{S}_{m}  \cup \{ j_{k} \} ,n_{\rm{sel}},m \right) +  f_{\bR,N_{\rm{RF}}} \left(\mathcal{S}_{l}  \setminus \{ j_{k} \} ,n_{\rm{sel}},l \right) $
\State $\mu_{{\rm{new}},i}  =  f_{\bR,N_{\rm{RF}}} \left( \mathcal{S}_{m}  \setminus \{ i_{k} \} ,n_{\rm{sel}},m \right) +  f_{\bR,N_{\rm{RF}}} \left( \mathcal{S}_{l}  \cup \{ i_{k} \} ,n_{\rm{sel}},l \right) $
\If {$\mu_{{\rm{new}},j}>\mu_{{\rm{new}},i}$, $ \mu_{{\rm{new}},j}>\mu_{\rm{current}}$, and $ m \neq 0$}
\State $\mathcal{S}_{m} \gets \mathcal{S}_{m}  \cup \{j_{k} \}$, $\mathcal{S}_{l} \gets \mathcal{S}_{l} \setminus \{ j_{k} \} $
\ElsIf{$\mu_{{\rm{new}},i}>\mu_{{\rm{new}},j}$, $ \mu_{{\rm{new}},i}>\mu_{\rm{current}}$, and $l \neq 0$}
\State $\mathcal{S}_{m} \gets \mathcal{S}_{m}  \setminus \{i_{k} \}$, $\mathcal{S}_{l} \gets \mathcal{S}_{l} \cup \{ i_{k} \} $
\EndIf

\EndIf
\EndFor 

\State Output: $\mathcal{S}_{1},\cdots,\mathcal{S}_{N_{\rm{RF}}}$
\end{algorithmic}
\end{algorithm}

We propose a practical algorithm using this approximate value of the largest singular value instead of the exact one. At the initial stage, the absolute values in the upper triangular part of $\bR$ matrix are sorted in descending order. Then, according to the sorted order, the following process is performed repeatedly.  If the selected element at each iteration stage is $|[\bR]_{i,j}|$, then the algorithm checks whether $i$-th antenna and $j$-th antenna are in the same subset or not. If they are in different subsets, the algorithm tries relocating one antenna to the subset that the other antenna belongs to, and calculates the metric, which is defined as the sum of the proposed approximate largest singular values of submatrices. Note that only at most two subsets can be changed at each stage while other subsets remain unchanged. Therefore, the singular values of other submatrices need not be recalculated, and thus the metric at each stage can be simplified as the sum of the two singular values. If the newly calculated metric is larger than the current metric, then the algorithm decides to relocate the antenna, and otherwise decides to maintain the current status. The pseudo code of the details in the proposed algorithm is shown in Algorithm 1. In Algorithm 1, the function $f_{\bR,N_{\rm{RF}}}(.)$ is defined as
\begin{equation}\label{f_fun_def}
f_{\bR,N_{\rm{RF}}} \left( \mathcal{S},n_{\rm{sel}},r \right) \triangleq  \Bigg\{ \begin{array}{ll} {0}, \;\; \textrm{if} \;\; |\mathcal{S}|=0 \;\; \textrm{or} \;\; \left\{ {n_{\rm{sel}}=N_{\rm{RF}}\;\; \textrm{and} \;\; r=0}\right \} \\ { \frac{1}{\left\vert{\mathcal{S}}\right\vert} \sum_{i \in \mathcal{S}}  \sum_{j \in \mathcal{S}} |[\bR]_{i,j}|}, \;\; \textrm{otherwise} \end{array},
\end{equation}
which indicates the approximate singular value of the covariance matrix of the antenna subset.
%\bcomment{The analysis on UPA case is not included yet. (Need to be updated)}

%%%%%%%%%%%%%%%%%%%%%%%%%%%%%%%%%%%%%%%%%%%%%%%%%%%%%%%%%%%%%%%%%%%%%%%%%%%%%%%%%%%%%%%%%%%%%%%%%%%%%%%%%%%%%
\section{Simulation Results} \label{sec:Results}
%%%%%%%%%%%%%%%%%%%%%%%%%%%%%%%%%%%%%%%%%%%%%%%%%%%%%%%%%%%%%%%%%%%%%%%%%%%%%%%%%%%%%%%%%%%%%%%%%%%%%%%%%%%%%

%\begin{figure}[!t]
%	\centerline{\resizebox{0.6\columnwidth}{!}{\includegraphics{Fig_JensenUB_mmWaveCh_r4-eps-converted-to.pdf}}}
%	%\centering
%	%\includegraphics[width=3.7in, height=3.1in]{Fig2.pdf}
%	\caption{Comparison between the exact value of the normalized spectral efficiency in \eqref{eq:norm_se_exact} and the approximate one in \eqref{eq:norm_se_approx}, with respect to the number of transmit antnenas when the number of receive antennas is fixed at two. }
%
%	\label{fig:Fig_Approx_ErgodicCapacity}
%\end{figure}

In this section, we first evaluate the performance of the proposed wideband hybrid precoding design in a mmWave frequency selective channel, and then present simulation results to demonstrate the performance of the proposed dynamic subarray algorithm with hybrid architectures.

In the simulations, we consider the channel model in  \sref{sec:Model}. The channel is modeled as a clustered channel where each cluster is composed of multiple subrays. The distributions of the paths' delay and azimuth/elevation angles are similar to that in the 3GPP 3D-MIMO channel model \cite{TR36873} and WINNER II SCM channel model \cite{WINNER2}. Considering multiple rays per cluster, we can rewrite the channel model in \eqref{eq:H_k} as 
%\begin{equation}\label{eq:CIR_model_modif}
%\bH_{\rm{CIR}}(\tau) = \sum_{c=1}^{N_{\rm{cluster}}} \sum_{s=1}^{N_{\rm{subray}}} \rho_{c,s}e^{j\omega_{c,s}} \delta \left( \tau - \tau_{c} \right) \ba_{{\rm{R}}} \left( \phi_{{\rm{R}},{c,s}}, \theta_{{\rm{R}},{c,s}} \right)  \ba^{*}_{{\rm{T}}} \left( \phi_{{\rm{T}},{c,s}}, \theta_{{\rm{T}},{c,s}} \right), 
%\end{equation}
%\begin{equation}\label{eq:H_k_model_modif}
%\bH[k] = \sum_{c=1}^{N_{\rm{cluster}}} \sum_{r=1}^{N_{\rm{subray}}}  \left( \alpha_{c,r} \sum_{d=0}^{D-1}  p( d T_{\rm{s}} - \tau_{c,r} ) e^{-\frac{j2\pi k  d }{K}} \right)  \ba_{{\rm{R}}} ( \phi_{{\rm{R}},c,r}, \theta_{{\rm{R}},c,r} )  \ba^{*}_{{\rm{T}}} ( \phi_{{\rm{T}},c,r}, \theta_{{\rm{T}},c,r} ), 
%\end{equation}
\begin{equation}\label{eq:H_k_model_modif}
\bH[k] = \sum_{c=1}^{N_{\rm{cluster}}} \sum_{r=1}^{N_{\rm{subray}}}  \alpha_{c,r} \omega_{ \tau_{c,r}}[k]   \ba_{{\rm{R}}} ( \phi_{{\rm{R}},c,r}, \theta_{{\rm{R}},c,r} )  \ba^{*}_{{\rm{T}}} ( \phi_{{\rm{T}},c,r}, \theta_{{\rm{T}},c,r} ), 
\end{equation}

\begin{figure}[t]
	\centering
	\subfigure[center][{IID Rayleigh channel model}]{
		\includegraphics[width=.63\columnwidth]{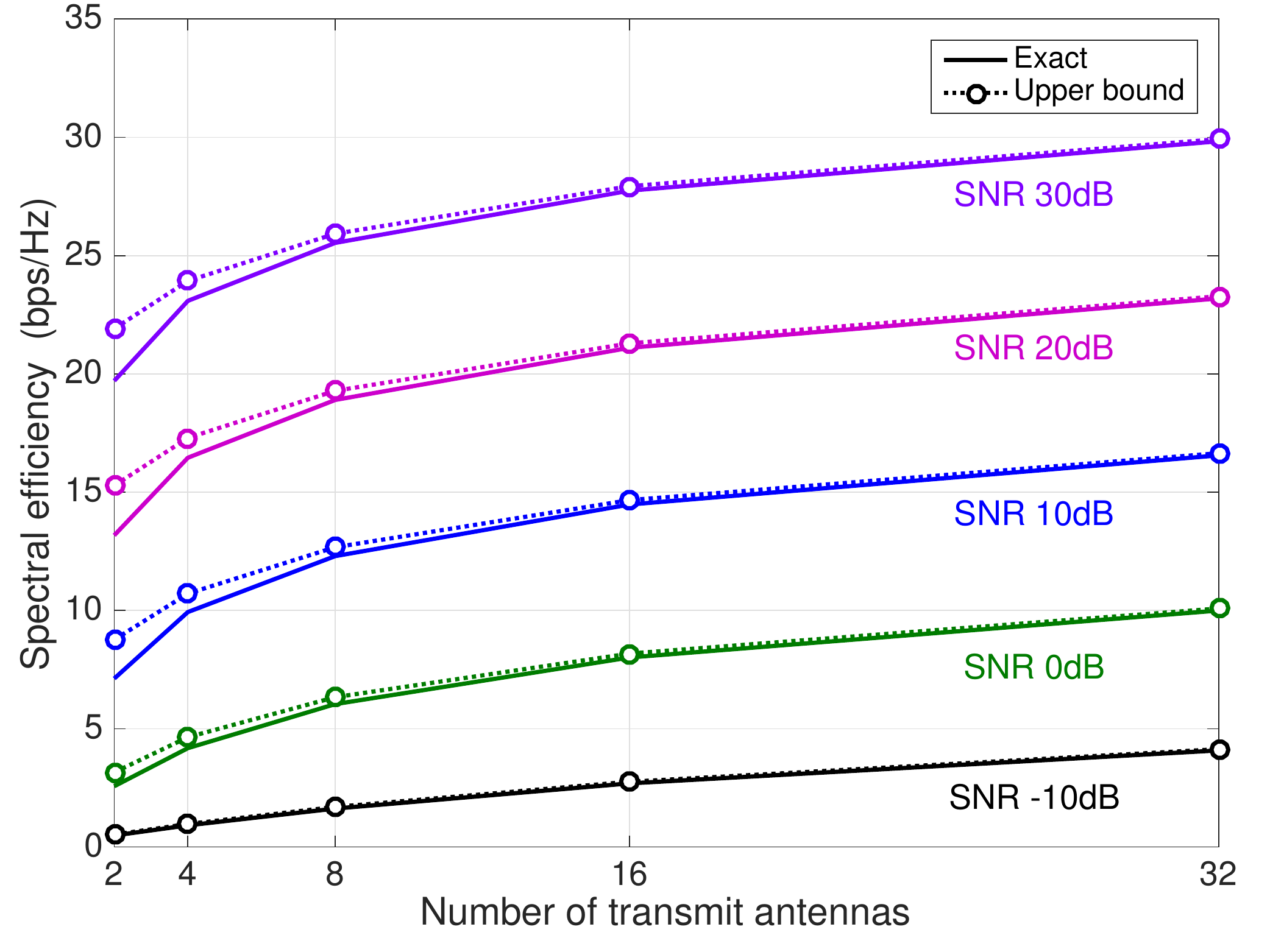}
		\label{fig:Fig_Approx_ErgodicCapacity_IIDch}}
	\subfigure[center][{mmWave channel model (8 clusters, 10 subrays)}]{
		\includegraphics[width=.63\columnwidth]{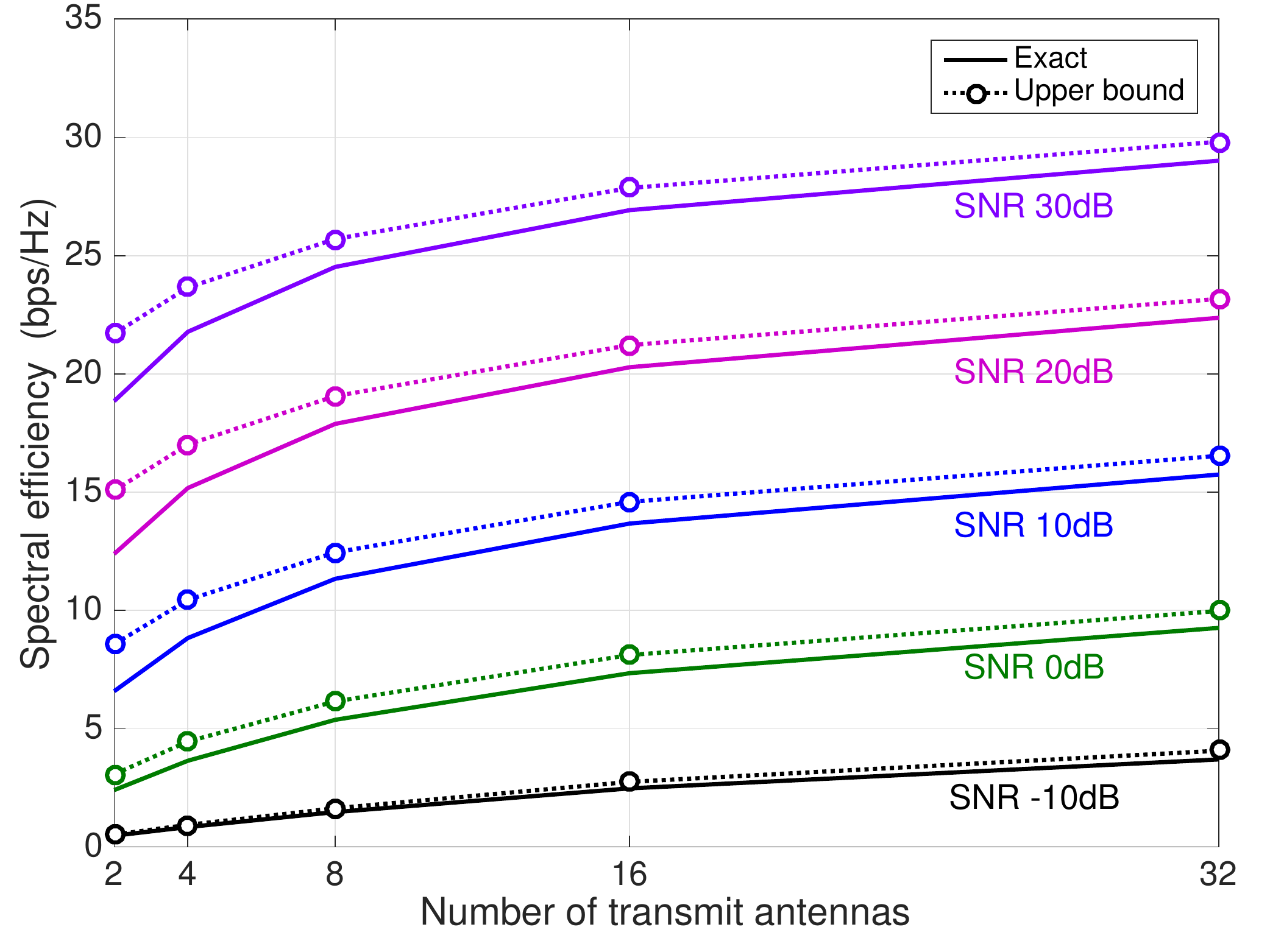}
		\label{fig:Fig_Approx_ErgodicCapacity_mmWaveCh}}
	\caption{Comparison between the exact value given in \eqref{eq:norm_se_exact} and the approximate value given in \eqref{eq:norm_se_approx}, which is Jensen's upper bound of \eqref{eq:norm_se_exact}. The number of receive antennas is fixed at two. IID Rayleigh channel model is assumed in (a), and the mmWave channel model is used in (b).}
	\label{fig:Fig_Approx_ErgodicCapacity}
\end{figure}

Unless otherwise mentioned, the adopted channel has $N_{\rm{cluster}}=8$ clusters whose center azimuth angles of arrival and  departure are uniformly distributed in $[-180^{\circ}, 180^{\circ}]$, and the center elevation angles of arrival and departure are uniformly distributed in $[-90^{\circ}, 90^{\circ}]$ when UPA is used in the simulation. Each cluster is composed of  $N_{\rm{subray}}=10$ subrays whose azimuth and elevation angles are assumed to be Laplacian distributed with angular spread of $5^{\circ}$ \cite{WINNER2}. Both ULA and UPA types are simulated, and the antenna spacing between antennas is $0.5\lambda$, where $\lambda$ is the signal wavelength. A raised-cosine filter with a roll-off factor one is adopted for the pulse shaping filter. The number of subcarriers $K$ is 4096, and the cyclic prefix length $D$ is assumed to be $K/4$ as in IEEE 802.11ad. All subrays within a cluster is assumed to have an identical delay such that $\tau_{c,1}=\cdots=\tau_{c,N_{\rm{subray}}}=\tau_{c}$. The cluster delay $\tau_{c}$ normalized to $T_{\rm{s}}$ is assumed to have a discrete uniform distribution in the cyclic prefix duration, $[0,D]$. The water-filling power control policy is used for all test cases.

\subsection{Evaluating the Relaxation of the Optimum Criterion } \label{subsec:eval_relax}

We used the relaxed optimum criterion in \eqref{eq:opt_criterion_relaxed} instead of the exact optimum criterion \eqref{eq:opt_criterion_equiv_onlyFrf}. \figref{fig:Fig_Approx_ErgodicCapacity} shows the exact value in the original problem 
 \begin{equation}\label{eq:norm_se_exact}
\frac{1}{KS} \sum_{k=1}^{K}  \sum_{s=1}^{S}   \log \left( 1 + \frac{ \lambda^{2}_{s} \left( \bH[k] \right)}{\sigma_{N}^2}  \right), 
\end{equation}
and the approximate value in the relaxed criterion 
\begin{equation}\label{eq:norm_se_approx}
   \log \left( 1 + \frac{1}{ KS} \sum_{k=1}^{K}  \sum_{s=1}^{S} \frac{\lambda^{2}_{s} \left( \mathbf{H} [k]  \right) }{\sigma_{N}^2} \right), 
\end{equation}
%\figref{fig:Fig_Approx_ErgodicCapacity} shows the exact value of the normalized spectral efficiency as
%\begin{equation}\label{eq:norm_se_exact}
%\frac{1}{KS} \sum_{k=1}^{K}  \sum_{s=1}^{S}   \log \left( 1 + \frac{ \lambda^{2}_{s} \left( \mathbf{H} \left[ k \right] \mathbf{F}_{\rm{RF}} (\bF^{*}_{\rm{RF}} \bF_{\rm{RF}})^{-\frac{1}{2}} \right)}{\sigma_{N}^2}  \right), 
%\end{equation}
%and the approximate value in the relaxed criterion as
%\begin{equation}\label{eq:norm_se_approx}
%   \log \left( 1 + \frac{1}{ KS} \sum_{k=1}^{K}  \sum_{s=1}^{S} \frac{\lambda^{2}_{s} \left( \mathbf{H} \left[ k \right] \mathbf{F}_{\rm{RF}} (\bF^{*}_{\rm{RF}} \bF_{\rm{RF}})^{-\frac{1}{2}}  \right) }{\sigma_{N}^2} \right), 
%\end{equation}
which is Jensen's upper bound of \eqref{eq:norm_se_exact}. 
%In \figref{fig:Fig_Approx_ErgodicCapacity}, $N_{\rm{RX}}$ is fixed at two, and $N_{\rm{RF}}$ is set to $\sqrt{N_{\rm{TX}}}$. The mmWave channel model in \eqref{eq:H_k_model_modif} is used with ULA type antennas at both BS and MS, and the unconstrained $\bF_{\rm{RF}}$ in \eqref{eq:opt_F_RF} is applied. 
%The relaxed objective function can be regarded as Jensen's upper bound of the average mutual information. 
It is well known that this bound is tight only at low SNR region. The bound, however, can be also tight even at high SNR if the number of transmit antennas is larger than the number of receive antennas.   
%Even at high SNR, the gap between the bound and the exact value can become small if the number of transmit antennas is larger than the number of receive antennas. 
\figref{fig:Fig_Approx_ErgodicCapacity_IIDch} shows the bound and the exact value in IID Rayleigh fading channel according to the number of transmit antennas when the number of receive antennas is two. 
%As can be seen in \figref{fig:Fig_Approx_ErgodicCapacity}, 
It is shown that the gap becomes smaller as the ratio of the number of transmit antennas to the number of receive antennas becomes larger.
%As the ratio of the number of transmit antennas to the number of receive antennas becomes larger, though, the gap between the bound and the exact value becomes smaller. 
%This makes this problem relaxation in \eqref{eq:opt_criterion_relaxed} work well in mmWave systems as large antenna arrays will normally be employed.    
%\figref{fig:Fig_Approx_ErgodicCapacity_IIDch} shows that this feature in IID Rayleigh block fading channel when the number of receive antennas is fixed at two. 
The gap is negligible when the number of transmit antennas is more than 16 even at high SNR. 
%If we consider the effective MIMO channel where the inputs are from the RF chains at the transmitter and the outputs are from RF chains at the receiver, then the number of antennas in the figure can be regarded as the number of RF chains.
\figref{fig:Fig_Approx_ErgodicCapacity_mmWaveCh} shows the results of the sparse mmWave channel case in \eqref{eq:H_k_model_modif}, where $N_{\rm{cluster}}=8, N_{\rm{subray}}=10$, and ULA type antennas are used at BS. These results indicate that the gap between the bound and the exact value is $\sim$1 bps/Hz when more than 16 antennas are deployed, which means that the relaxed optimization problem in \eqref{eq:opt_criterion_relaxed} is a reasonable approximation of the original problem in \eqref{eq:opt_criterion_equiv_onlyFrf} for large MIMO mmWave systems.
%This result indicates that the relaxed optimization problem in \eqref{eq:opt_criterion_relaxed} is an reasonable approximation of the original problem in \eqref{eq:opt_criterion_equiv_onlyFrf} for large mmWave MIMO systems.

\subsection{Wideband Hybrid Precoding over Frequency Selective Channels}\label{subsec:WidebandHybridPrecodingOverFrequencySelectiveChannel}

%\begin{figure}[!t]
%	\centerline{\resizebox{0.65\columnwidth}{!}{\includegraphics{Fig_DSA_3x3UPA3RFatBS_2ULAatMS_r1-eps-converted-to.pdf}}}
%\centering
%\includegraphics[width=3.7in, height=3.1in]{Fig2.pdf}
%	\caption{Comparion between the ideal digitial baseband precoding, the hybrid precoding with a fully-connected structure, the hybrid precoding with the proposed dynamic subarray structure, and the hybrid precoding with several fixed subarray stuctures in the UPA case. The parameters are the same as in \figref{fig:Fig_SNRvsRate_ULA} except that 9 antennas in the form of 3x3 UPA and 3 RF chains are equipped at the transmitter. The receiver has 2 antennas(ULA) and 2 RF chains as in \figref{fig:Fig_SNRvsRate_ULA} }\label{fig:Fig_SNRvsRate_9UPA_3RF_360}
%\end{figure}

In \figref{fig:Fig_SNRvsRate_NumRF_Rayleigh} and \figref{fig:Fig_SNRvsRate_NumRF_Sparse}, we evaluate the performance of the proposed wideband hybrid precoding design for SU-MIMO over a frequency selective channel in a fully-connected structure. The figures show the average mutual information per subcarrier according to SNR when the wideband hybrid precoding in \sref{sec:Fully} are used in the case of 16 transmit antennas and 4 receive antennas. The number of RF chains at the receiver is 4, and fully-digitalized baseband combining is used. The number of RF chains at the transmitter is 1, 2, 4, or 8. In \figref{fig:Fig_SNRvsRate_NumRF_Rayleigh}, the channel per subcarrier and per antenna is modeled as IID Rayleigh channel, which is an extreme case of an ideal rich scattering environment. The results show that there is a substantial loss from the fully-digitalized baseband precoding case even when eight RF chains are used in the wideband hybrid precoding. This, however, is not the case when the sparse mmWave channel is considered. \figref{fig:Fig_SNRvsRate_NumRF_Sparse} shows the wideband hybrid precoding performance when the mmWave channel model in \eqref{eq:H_k_model_modif} is adopted, with $N_{\rm{cluster}}=8$ and $N_{\rm{subray}}=10$. If each cluster has only one ray and the number of clusters is less than or equal to the number of RF chains, the performance of the proposed wideband hybrid precoding is the same as that of the fully-digital precoding as discussed in \sref{sec:Fully}. Even when the channel clusters have multiple subrays with angle spread $5^{\circ}$, \figref{fig:Fig_SNRvsRate_NumRF_Sparse} shows that the performance gap between the proposed hybrid precoding and the fully-digital solution is negligible when eight RF chains.

\begin{figure}[!t]
	\centering
	\subfigure[center][{IID Rayleigh channel model}]{
		\includegraphics[width=.63\columnwidth]{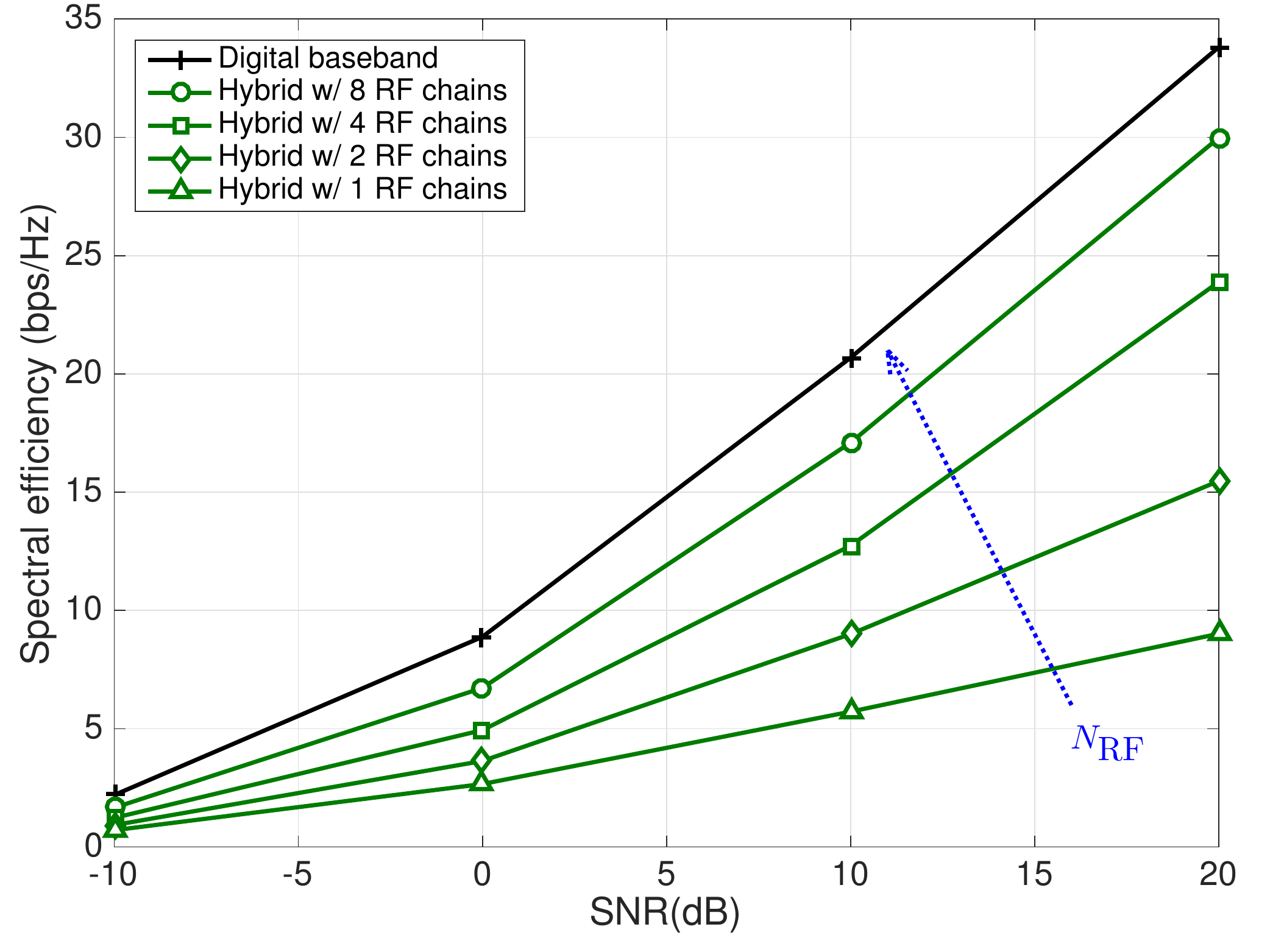}
		\label{fig:Fig_SNRvsRate_NumRF_Rayleigh}}
	\subfigure[center][{mmWave channel model (8 clusters, 10 subrays)}]{
		\includegraphics[width=.63\columnwidth]{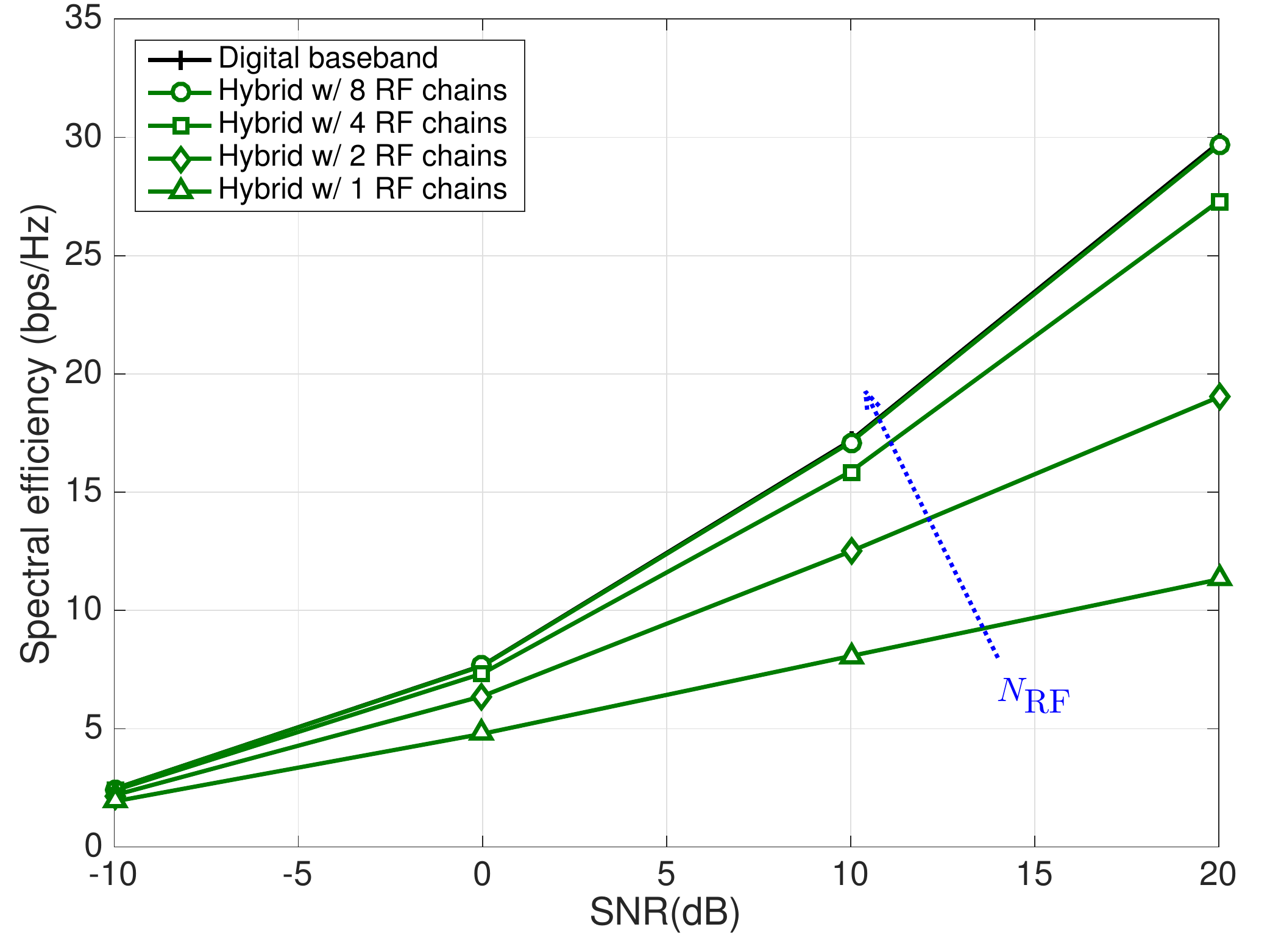}
		\label{fig:Fig_SNRvsRate_NumRF_Sparse}}
	\caption{Comparison between the fully-digitalized baseband precoding and the hybrid precoding with 1, 2, 4, or 8 RF chains in a fully-connected structure. The trasmitter uses 16 antennas (ULA), and the reciever uses 4 antennas (ULA). In (a), the channel at each subcarreir is modeled as IID Rayleigh channel. The mmWave channel model used in (b) is the same as in \figref{fig:Fig_Approx_ErgodicCapacity_mmWaveCh}.}
	\label{fig:Fig_SNRvsRate_NumRF}
\end{figure}

\subsection{Wideband Hybrid Precoding with Dynamic Subarray Structures}

In this subsection, we evaluate the performance of the proposed algorithm in the subarray structure. The proposed dynamic subarray technique is compared with several fixed subarray types as well as the fully-connected hybrid precoding and the fully-digitalized baseband precoding. For the dynamic subarray architecture, the proposed greedy algorithm is compared with the optimal exhaustive search algorithm. For comparison, we also simulate another simple technique of dynamic subarrays that selects the best subarray architecture among a predefined fixed subarray types. The simulations are conducted in various channel environments to establish the dependence of the dynamic subarray gain on channel parameters. In addition to evaluating the dynamic subarrays, we also establish which is the best fixed subarray structure and which channel parameters affect the decision of the best structure. 
 
\begin{figure}[!t]
	\centering
	\subfigure[center][{Fixed subarray types for simulations: 9 antennas (1x9 ULA) and 3 RF chains.}]{
		\includegraphics[width=.8\columnwidth]{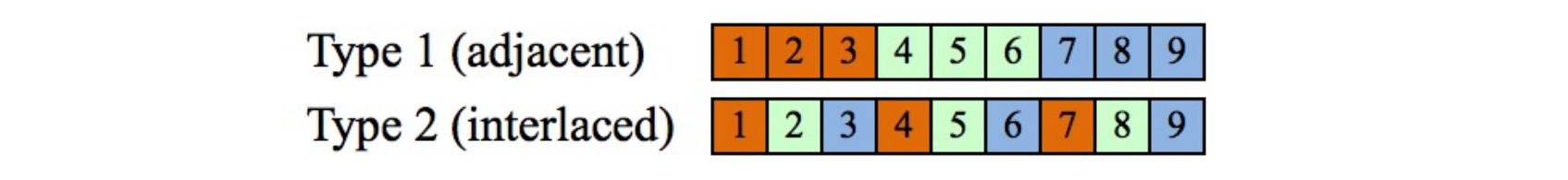}
		\label{fig:Fig_FixedSubarryType_9ULA_3RF}}
	\subfigure[center][{Spectral efficiency vs. SNR.}]{
		\includegraphics[width=.63\columnwidth]{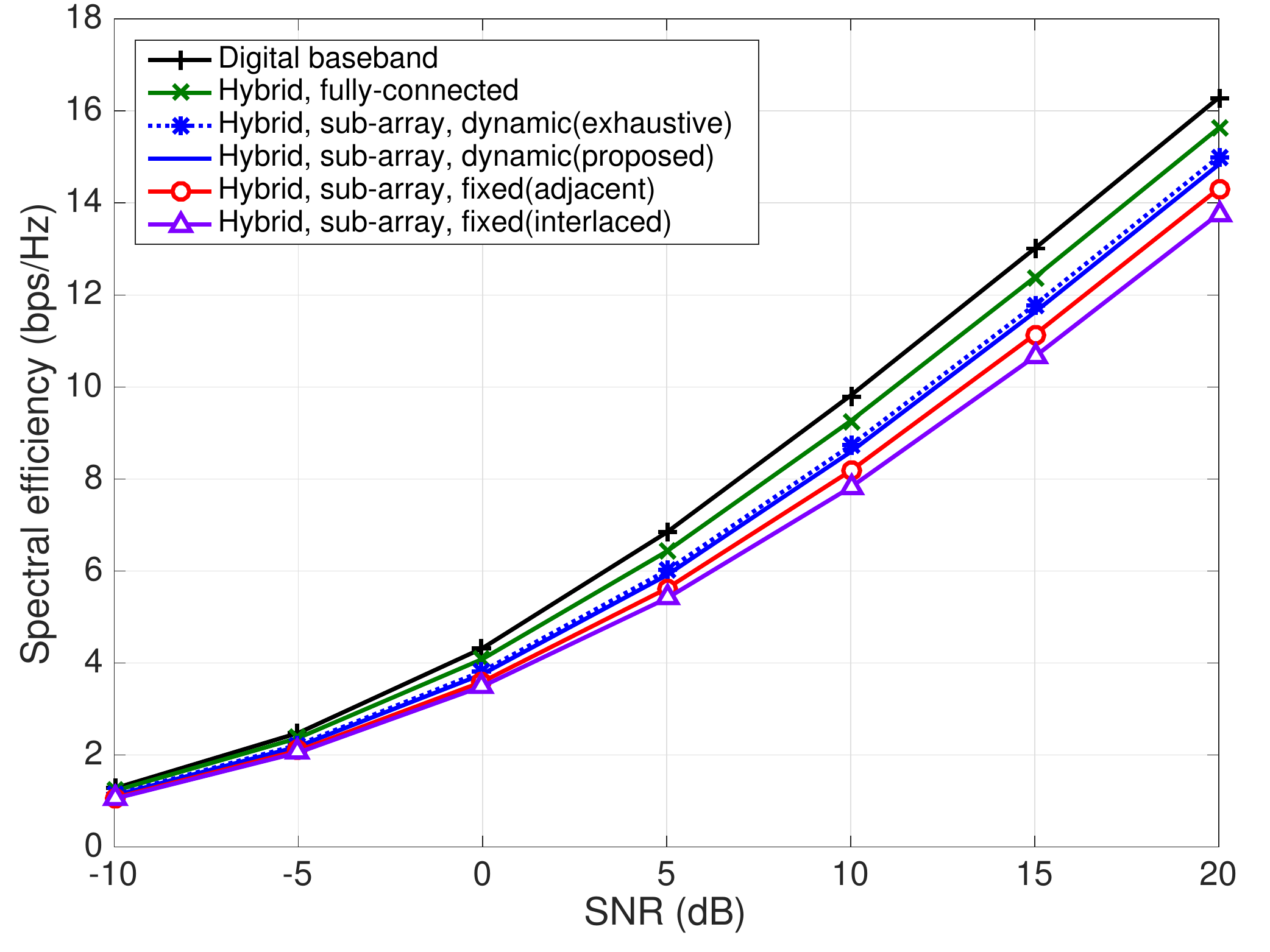}
		\label{fig:Fig_SNRvsRate_9ULA_3RF_360}}
	\caption{Comparion among various hybrid architectures when 9 antennas (ULA) and 3 RF chains are equipped at BS and 2 antennas (ULA) and 2 RF chains are equipped at MS. }
	\label{fig:Fig_SNRvsRate_9ULA_3RF}
\end{figure}

\textbf{Gain over fixed subarrays in ULA systems:} \figref{fig:Fig_SNRvsRate_9ULA_3RF} shows the comparison of various precoding techniques when the base station has 9 antennas (ULA) and 3 RF chains and the mobile station has 2 antennas (ULA) and 2 RF chains. For comparison, two fixed subarray types are used in the simulation as described in \figref{fig:Fig_FixedSubarryType_9ULA_3RF}: an adjacent type and an interlaced type. \figref{fig:Fig_SNRvsRate_9ULA_3RF_360} shows that, in the dynamic subarray architecture, the performance of the proposed algorithm is close to the optimal exhaustive search case even with a much lower complexity. The results also indicate that the adjacent type is better than the interlaced one among two fixed subarray structures. This is because the largest singular value of each adjacent type is larger than that of a interlaced type when the channel is correlated.

\begin{figure}[!t]
	\centering
	\subfigure[center][{Fixed subarray types for simulations: 64 antennas (8x8 UPA) and 4 RF chains.}]{
		\includegraphics[width=1.0\columnwidth]{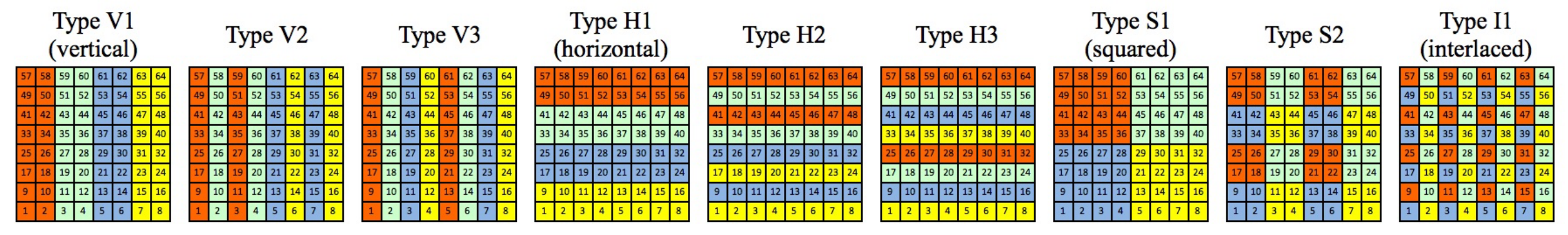}
		\label{fig:FixedSubarrayTypes_64UPA_4RF_360}}
	\subfigure[center][{Spectral efficiency vs. SNR.}]{
		\includegraphics[width=.63\columnwidth]{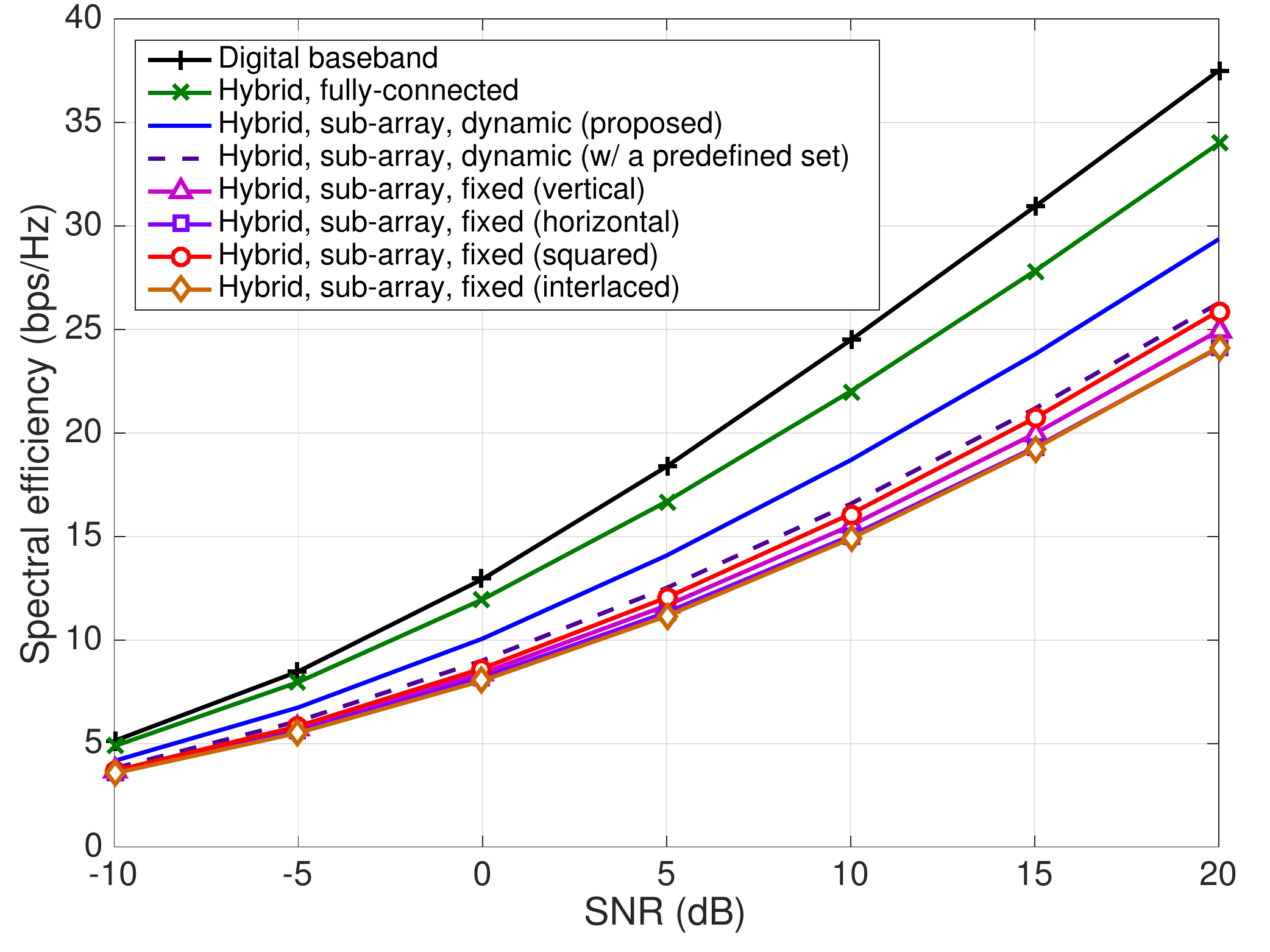}
		\label{fig:Fig_SNRvsRate_64UPA_4RF_360}}
	\subfigure[center][{Histogram of the selected fixed subarrays in a predefined set}]{
		\includegraphics[width=.63\columnwidth]{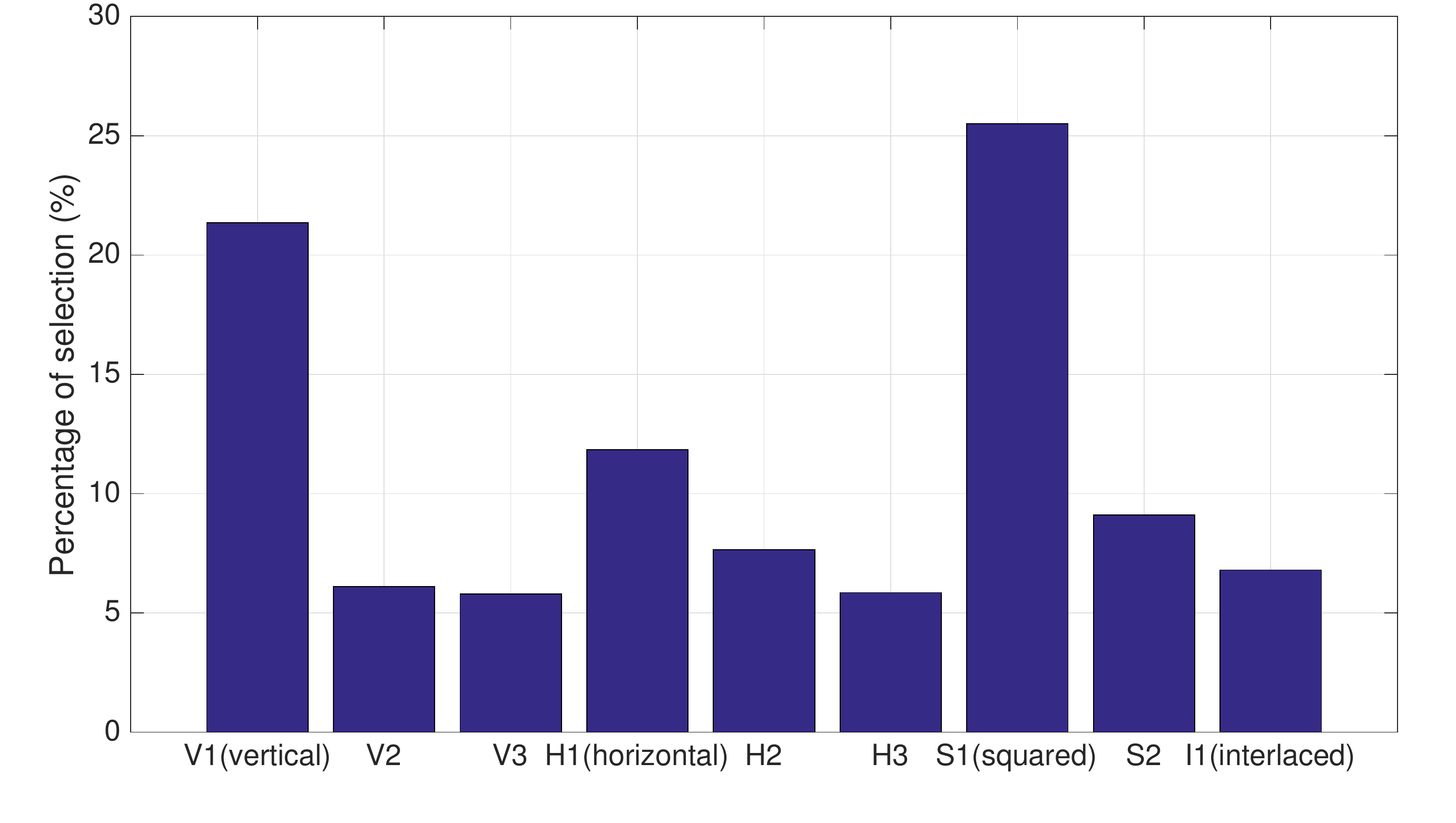}
		\label{fig:Fig_Hist_64UPA_4RF_360}}
	\caption{Comparion among various hybrid architectures when 64 antennas (8x8 UPA) and 4 RF chains are equipped at BS and 4 antennas (2x2 UPA) and 4 RF chains are equipped at MS. }
	\label{fig:Fig_SNRvsRate_64UPA_4RF}
\end{figure}

\textbf{Gain over fixed subarrays in UPA systems:} \figref{fig:Fig_SNRvsRate_64UPA_4RF} shows a simulation result when 64 antennas (8x8 UPA) and 4 RF chains are used at the base station and 4 antennas (2x2 UPA) and 4 RF chains are used at the mobile station. Nine fixed subarray types are used in the simulation as shown in \figref{fig:FixedSubarrayTypes_64UPA_4RF_360}. The optimal exhaustive search algorithm for the dynamic subarrays was not simulated because its computational complexity is too high in this case. Instead, a simple dynamic subarray algorithm that selects the best subarray structure in a predefined set, which consists of the nine fixed subarray types in \figref{fig:FixedSubarrayTypes_64UPA_4RF_360}, was simulated for comparison. Even though the dynamic subarray technique with a predefined set outperforms any fixed subarray types, this naive dynamic subarray technique is considerably outperformed by the proposed algorithm. This is because the simple dynamic algorithm selects the best subarray type among nine fixed subarray types while the proposed algorithm can decide the best subarray type among all possible types, whose total number is $1.4 \times 10^{37}$. Compared to the results in \figref{fig:Fig_SNRvsRate_9ULA_3RF}, we can see that the gain of the proposed dynamic algorithm becomes higher as the number of antennas and RF chains gets larger. Apart from the dynamic subarrays, the figures shows the information about the best structure when a fixed subarray structure is applied. As can be seen in \figref{fig:Fig_SNRvsRate_64UPA_4RF_360}, the squared type is the best structure among the fixed subarray structures, and the vertical type is the second. This trend is consistent with  the results in \figref{fig:Fig_Hist_64UPA_4RF_360} that shows the selection ratio of fixed subarray structures when a simple dynamic algorithm with a predefined set is used.

\textbf{Impact of channel parameters:} 
%Even though \figref{fig:Fig_SNRvsRate_64UPA_4RF} shows that the proposed dynamic subarray technique has a significant gain compared to fixed subarrays and that the squared fixed subarray type is the best among fixed subarrays, this is not always the case. Both the dynamic subarray gain and the best subarray structure depend on the channel environment. 
The best fixed subarray structure as well as the dynamic subarray gain depends the channel environment. In particular, the distributions of azimuth and elevation angles of channel paths play an important role, as the angle distributions affect the largest singular value of each subarray.  
%\figref{fig:Fig_SNRvsRate_64UPA_4RF} shows a simulation result when 64 antennas (ULA) and 4 RF chains are used at the base station and 2 antennas (ULA) and 2 RF chains are used at the mobile station. In this case, the computation time of the exhaustive search is so long that its simulation could not be performed. The comparison between the results in \figref{fig:Fig_SNRvsRate_9ULA_3RF_360} and \figref{fig:Fig_SNRvsRate_64ULA_4RF_360} indicates that the gain of the proposed dynamic subarray technique becomes higher as the number of antennas gets larger. 
%
%Performance evaluation with UPAs is shown in \figref{fig:Fig_SNRvsRate_UPA}. In \figref{fig:Fig_SNRvsRate_9UPA_3RF_360}, the BS uses 3 RF chains and a 9 antenna element UPA with a $3\times 3$ dimension, and the MS uses 2 RF chains and a 2 antenna element ULA. The achievable rate of the proposed dynamic subarray is close to that of the optimal exhaustive search algorithm as in \figref{fig:Fig_SNRvsRate_9ULA_3RF_360}. As in the ULA case, the gain of the dynamic subarray type becomes large when the number of antennas is large. \figref{fig:Fig_SNRvsRate_81UPA_9RF_360} presents the performance results when 81 antennas with a $9 \times 9$ dimension and 9 RF chains are used at the base station. Among four fixed subarray 
%types, the squared type is the best, which is, however, outperformed by the dynamic subarray structure. 
The azimuth angles and elevation angels can be confined within some range in some cell deployment scenarios. For example, the range of incoming azimuth angles can be restricted in a 3-sectorized cell scenario where sector antennas with directional antenna gain are equipped. In addition, the range of azimuth angles can be different from that of elevation angles. For example, many outdoor scenarios are usually assumed to have a smaller range of elevation angles than that of azimuth angles \cite{WINNER2}.  

\begin{figure}[!t]
	\centering
	\subfigure[center][{Center azimuth angles are distributed in $[-\phi_{\rm{max}},\phi_{\rm{max}}]$}]{
		\includegraphics[width=.63\columnwidth]{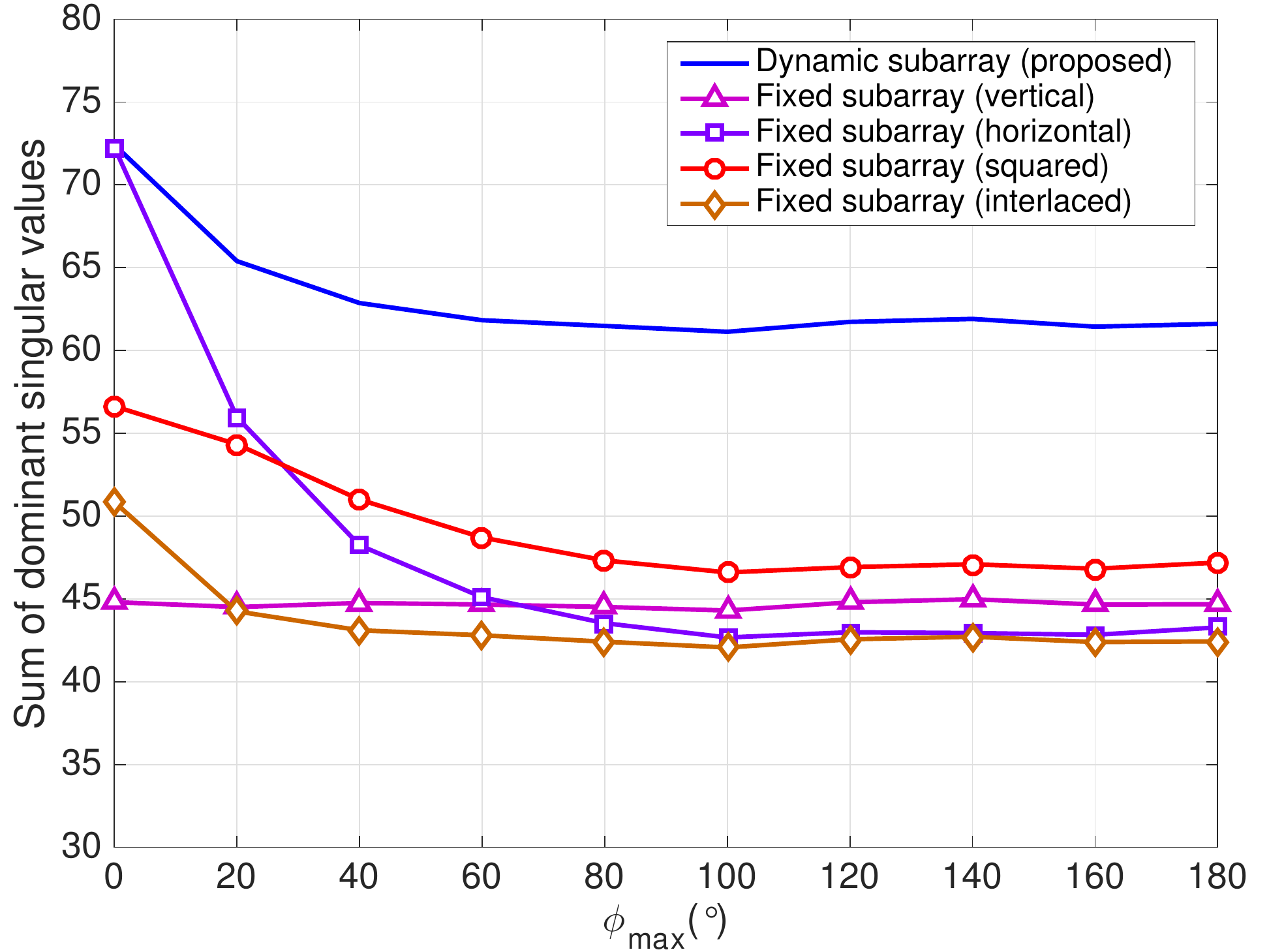}
		\label{fig:Fig_AzimuthAngleRange}}
	\subfigure[center][{Center elevation angles are distributed in $[-\theta_{\rm{max}},\theta_{\rm{max}}]$.}]{
		\includegraphics[width=.63\columnwidth]{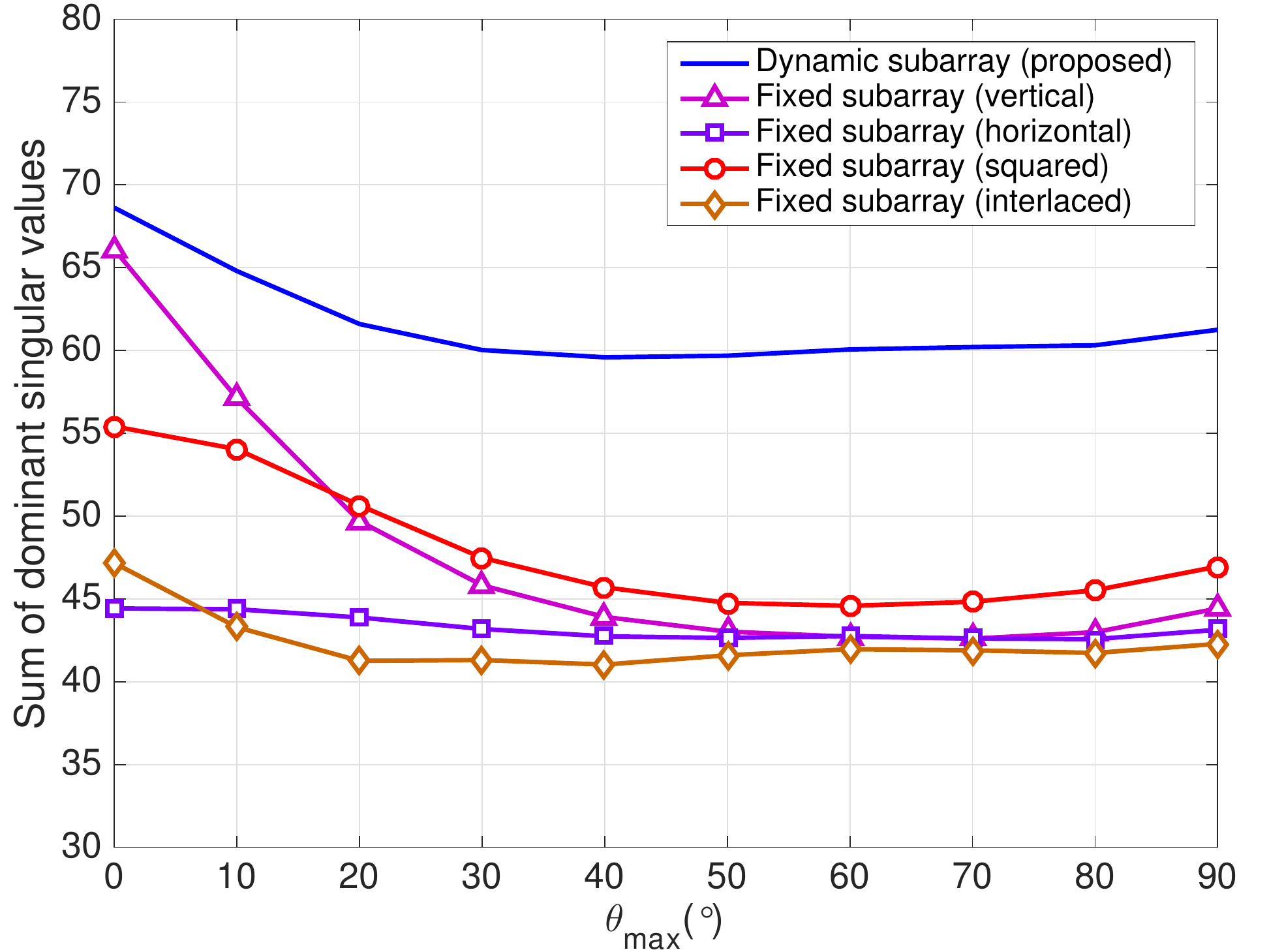}
		\label{fig:Fig_ElevationAngleRange}}
	\caption{The objective function in \eqref{eq:opt_criterion_dynamic_subarray} vs. range of angles. The BS uses 81 antennas (9x9 UPA) and 9 RF chains, and the MS uses 4 antenna (2x2 UPA) and 4 RF chains. The angle spread within a cluster is $5^{\circ}$, and SNR is 10dB. In (a), The center azimuth angles per cluster are uniformly distributed within $[-\phi_{\rm{max}},\phi_{\rm{max}}]$, and the center elevation angles per cluster are uniformly distributed within $[-90^{\circ},90^{\circ}]$. In (b), The center azimuth angles per cluster are uniformly distributed within $[-180^{\circ},180^{\circ}]$, and the center elevation angles per cluster are uniformly distributed within $[-\theta_{\rm{max}},\theta_{\rm{max}}]$. }	
	\label{fig:Fig_AngleRangePerCluster}
\end{figure}

\figref{fig:Fig_AngleRangePerCluster} shows the influence of the azimuth and elevation angle range on the objective function in \eqref{eq:opt_criterion_dynamic_subarray}, which is the sum of the dominant singular values of each subarray. In \figref{fig:Fig_AzimuthAngleRange}, the center azimuth angles per cluster are assumed to be uniformly distributed within $[-\phi_{\rm{max}},\phi_{\rm{max}}]$, so the maximum azimuth difference from the antenna boresight angle is limited to $\phi_{\rm{max}}$. The center elevation angles are assumed to be uniformly distributed within $[-90^{\circ},90^{\circ}]$, which means that there is no restriction on the elevation angle range. The figure shows that the gain of the dynamic subarray technique increases as the range of angles becomes wider. The figure also demonstrates that the best fixed subarray structure varies according to the angle range. The horizontal fixed type structure outperforms other fixed types when the azimuth angles are confined within a small range. The main reason is due to the difference in the range of angles. If the range in the azimuth angles is narrower than the range in the elevation angles, the largest singular values of the covariance channel matrix of each horizontal row is larger than that of each vertical column, and thus the horizontal fixed type structure has a higher value of the objective function in  \eqref{eq:opt_criterion_dynamic_subarray} than the vertical fixed type and others. The squared fixed type, however, becomes the best among fixed subarray types as the range of azimuth angles becomes bigger. This is because the horizontal domain and the vertical domain have a similar level of correlation, which enables the squared fixed type subarray structure to have the largest singular value due to the smallest distances between antennas. A similar phenomenon occurs when the elevation angles have a limited range. \figref{fig:Fig_ElevationAngleRange} shows that the vertical fixed type outperforms other fixed types when the range of elevation angles is small and the squared fixed type is the best at larger ranges as  \figref{fig:Fig_AzimuthAngleRange}.

%In addition to the angle range of each cluster, the angle spread within each cluster also has an effect on the performance as shown in \figref{fig:Fig_AngleSpreadWithinCluster}. As the angle spread becomes bigger, the channel environment is closer to a rich scattering case.  On the contrary, the small angle spread indicates more sparse channel. The figure illustrates that the proposed dynamic structure has a higher gain as the angle spread becomes smaller. The best fixed subarray structure, however, is not significantly affected by the angle spread in this case. Another finding from \figref{fig:Fig_AngleSpreadWithinCluster} besides the findings about the subarray structure is that the spectral efficiency of the fully-connected hybrid precoding is closer to that of the fully-digitalized baseband precoding as the angle spread becomes smaller. This result is consistent with  \sref{subsec:WidebandHybridPrecodingOverFrequencySelectiveChannel} in the sense that the sparse mmWave channel is favorable to the hybrid precoding.

\begin{figure}[!t]
	\centerline{\resizebox{0.64\columnwidth}{!}{\includegraphics{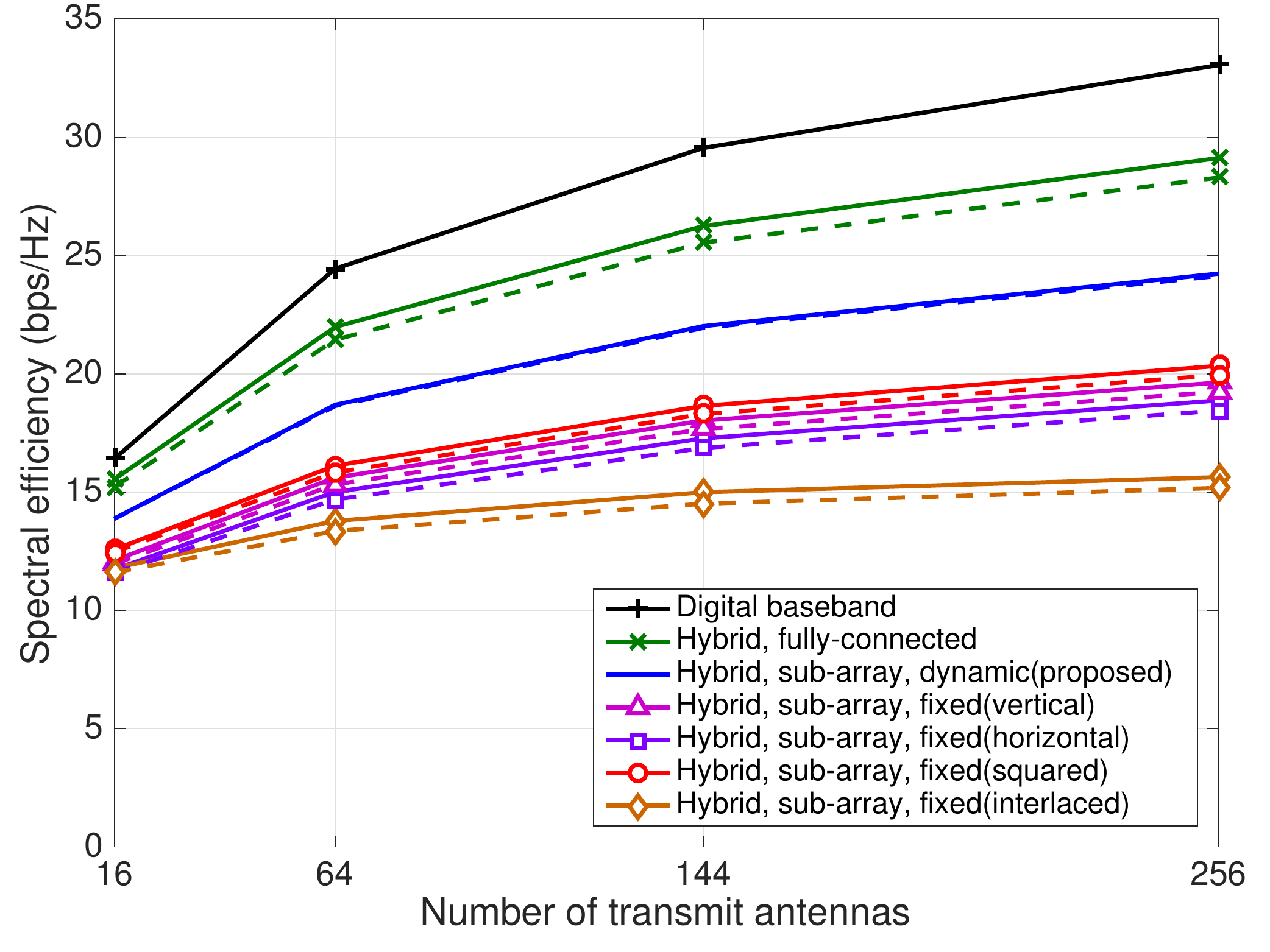}}}
	%\centering
	%\includegraphics[width=3.7in, height=3.1in]{Fig2.pdf}
	\caption{Spectral efficiency vs. number of transmit antennas. The number of   RF chains at BS is fixed at four. The MS uses 4 antenna (2x2 UPA) and 4 RF chains. The center azimuth angles per cluster are uniformly distributed in $[-180^{\circ},180^{\circ}]$, and the center elevation angels per cluster are uniformly distributed in $[-90^{\circ},90^{\circ}]$, the angles per subray within a cluster are Laplacian distributed with 5$^{\circ}$, and SNR is 10dB. The dashed curves indicate the case of RF precoding with only phase shifters under the constant modulus constraint while the solid ones represent the case of RF precoding without the constraint. }
	\label{fig:Fig_NumTx_NumRF_wConstantModulus}
\end{figure}

\textbf{Performance with different antenna array sizes:} The gain of proposed dynamic subarray structure also varies with the number of antennas as can be seen in \figref{fig:Fig_NumTx_NumRF_wConstantModulus}. In the figure, the MS has 4 antennas (2x2 UPA) and 4 RF chains, and the BS has 16, 64, 144, or 256 antennas (4x4, 8x8, 12x12, or 16x16 UPA). The number of RF chains at BS is assumed to be fixed at four. The figure also shows the impact of the phase shifter constraint. The solid curves represent  the unconstrained RF precoding, and the dashed curves indicate the constrained RF precoding with phase shifters in the analog RF precoding. The figure shows that the loss from using the phase shifters in the analog domain is not substantial. Apart from the phase shifter constraint issue, \figref{fig:Fig_NumTx_NumRF_wConstantModulus} also shows that the gain of the proposed dynamic structure becomes higher as the number of transmit antennas becomes larger.

%%%%%%%%%%%%%%%%%%%%%%%%%%%%%%%%%%%%%%%%%%%%%%%%%%%%%%%%%%%%%%%%%%%%%%%%%%%%%%%%%%%%%%%%%%%%%%%%%%%%%%%%%%%%%
\section{Conclusions} \label{sec:Conclusion}
%%%%%%%%%%%%%%%%%%%%%%%%%%%%%%%%%%%%%%%%%%%%%%%%%%%%%%%%%%%%%%%%%%%%%%%%%%%%%%%%%%%%%%%%%%%%%%%%%%%%%%%%%%%%%
%In this paper, we developed hybrid analog/digital precoding design for  wideband mmWave systems with frequency selectivity. First, we developed a closed-form solution that provides the optimal hybrid analog/digital precoders to a relaxation of the original rate maximization problem. This solution is developed for both fully-connected and partially-connected hybrid architectures. The developed solution gave useful tools to how the subarray structure should be designed. Based on this solution, the antenna should be associated to the subarrays that maximize the sum of the dominant eigenvalues of the channel covariance matrix. This insight led us to develop a low-complexity yet efficient algorithm that dynamically construct the hybrid architectures to achieve high spectral efficiency. Simulation results showed that the proposed hybrid precoding solution perform well compared to the fully-digital upper bounds in wideband mmWave systems. These results also show the spectral efficiency advantage of the developed dynamic subarray hybrid precoding solution over the fixed subarray structure in different system and channel configurations.  

In this paper, we developed hybrid analog/digital precoding design for wideband mmWave MIMO-OFDM systems with frequency selectivity. Considering a relaxation of the mutual information maximization problem, we derived a near-optimal closed-form solution for both fully-connected and partially-connected hybrid architectures. Simulation results showed that the spectral efficiency of the proposed wideband hybrid precoding designs approaches that obtained with fully-digital precoding. Inspired by the developed closed-form solution, we explored the potential spectral efficiency gain if the antenna subarrays can be adaptively adjusted according to the large channel statistics. For that, we first developed a criterion for constructing these subarrays, and used to design an antenna partitioning algorithm. One insight, drawn from the developed criterion, is that forming each subarray with more correlated antenna elements normally leads to an efficient subarray structure. Simulation results showed that the achievable spectral efficiency by dynamic subarrays outperforms that of fixed subarray architectures. For future work, it would be interesting to evaluate the trade-off between the achieved spectral efficiency and the consumed energy of the dynamic subarray structure, and compare it with the fully-connected and the fixed-subarray architectures.

%-%%%%%%----%%%%%%%----------------------------------------
\appendices
\section{}\label{app:prop_bounds}

\begin{proof}[Proof of Proposition \ref{prop:bounds}]
Let the diagonal elements of $\bR_{\mathcal{S}}$ be normalized to one and $|\mathcal{S}|$ be $n$. Then, \eqref{eq:m_and_s} becomes
%\begin{small}
\begin{equation}\label{eq:m_and_s_approx}
m=1, \;\;\; s=\left( \frac{1}{n} ||\bR_{\mathcal{S}}||^2_{\rm{F}} - 1  \right)^{\frac{1}{2}}=\left( \frac{2}{n} \sum_{i=1}^{n}\sum_{j>i}^{n} |[\bR_{\mathcal{S}}]_{i,j}|^2  \right)^{\frac{1}{2}}.
\end{equation}
%\end{small}
Define $\gamma_{k}$ for $k=1,\dots,\frac{n(n-1)}{2}$ to be
%\begin{small}
\begin{equation}\label{eq:gamma_k_def}
\begin{split}
\gamma_{1} =  |[\bR_{\mathcal{S}}]_{1,2}|, \gamma_{2} =  |[\bR_{\mathcal{S}}]_{1,3}|, \cdots, \gamma_{n-1} &=  |[\bR_{\mathcal{S}}]_{1,n}|, \\
\gamma_{n} =  |[\bR_{\mathcal{S}}]_{2,3}|, \cdots , \gamma_{2n-3} &=  |[\bR_{\mathcal{S}}]_{2,n}|, \\
\ddots \;\;\;\;\;\; & \vdots \\
\gamma_{\frac{n(n-1)}{2}} &=  |[\bR_{\mathcal{S}}]_{n-1,n}|.
\end{split}
\end{equation}
%\end{small}
Then, the lower and upper bounds in \eqref{eq:upper_and_lower_bound_def} can be rewritten as
%\begin{small}
\begin{equation}\label{eq:lower_and_upper_bound_def_approx}
\begin{split}
\lambda_{1, \rm{LB}} \left( \bR_{\mathcal{S}} \right)&=1+\left( \frac{2}{n(n-1)} \sum_{i=1}^{n}\sum_{j>i}^{n} |[\bR_{\mathcal{S}}]_{i,j}|^2  \right)^{\frac{1}{2}}= 1+\left( \frac{2}{n(n-1)} \sum_{k=1}^{n(n-1)/2} \gamma_k^2  \right)^{\frac{1}{2}},\\
\lambda_{1, \rm{UB}} \left( \bR_{\mathcal{S}} \right)&= 1+ \left( \frac{2(n-1)}{n} \sum_{i=1}^{n}\sum_{j>i}^{n} |[\bR_{\mathcal{S}}]_{i,j}|^2  \right)^{\frac{1}{2}}= 1+\left( \frac{2(n-1)}{n} \sum_{k=1}^{n(n-1)/2} \gamma_k^2  \right)^{\frac{1}{2}}.
\end{split}
\end{equation}
%\end{small}
The approximate value of the largest singular value in \eqref{eq:approx_singular_value} can be represented as
%\begin{small}
\begin{equation}\label{eq:approx_singular_value_rewrite}
\hat{\lambda}_{1} \left( \bR_{\mathcal{S}} \right) =1+ \frac{2}{n}\sum_{i=1}^{n}\sum_{j>i}^{n} |[\bR_{\mathcal{S}}]_{i,j}| = 1+  \frac{2}{n}\sum_{k=1}^{n(n-1)/2} \gamma_k
\end{equation}
%\end{small}
The relation between the lower bound of the exact value and the approximate value is 
%\begin{small}
\begin{equation}\label{eq:proof_on_lower_bound}
\frac{ \lambda_{1, \rm{LB}}\left( \bR_{\mathcal{S}} \right) - 1}{\hat{\lambda}_{1} \left( \bR_{\mathcal{S}} \right)-1}  = \left( \frac{n}{2(n-1)} \cdot \frac{\sum_{k=1}^{n(n-1)/2} \gamma_k^2}{\left(\sum_{k=1}^{n(n-1)/2} \gamma_k \right)^2} \right)^{\frac{1}{2}} \leq 1 \;\; \textrm{for} \;\; n \geq 2
\end{equation}
%\end{small}
because
\begin{equation}\label{eq:proof_on_lower_bound_sub1}
\frac{n}{2(n-1)} \leq 1 \;\; \textrm{for} \;\; n \geq 2
\end{equation}
and
%\begin{small}
\begin{equation}\label{eq:proof_on_lower_bound_sub2}
\begin{split}
\left(\sum_{k=1}^{n(n-1)/2} \gamma_k \right)^2 &= \sum_{k=1}^{n(n-1)/2} \gamma_k^2 +2 \sum_{k=1}^{n(n-1)/2}  \sum_{m>k}^{n(n-1)/2}  \gamma_k \gamma_m\\
&\geq   \sum_{k=1}^{n(n-1)/2} \gamma_k^2,
\end{split}
\end{equation}
%\end{small}
where the inequality stems from the fact that $\gamma_k \geq 0, \forall k$ by definition in \eqref{eq:gamma_k_def}. From \eqref{eq:proof_on_lower_bound}, we can conclude that
%\begin{small}
\begin{equation}\label{eq:proof_on_lower_bound_final}
\hat{\lambda}_{1} \left( \bR_{\mathcal{S}} \right) \geq  \lambda_{1, \rm{LB}} \left( \bR_{\mathcal{S}} \right)   \;\; \textrm{for} \;\; n \geq 2 .
\end{equation}
%\end{small}
Now, consider the relationship between the upper bound of the exact and approximate singular values. This ratio between the two values can be written as
%\begin{small}
\begin{equation}\label{eq:proof_on_upper_bound}
\frac{ \lambda_{1, \rm{UB}}\left( \bR_{\mathcal{S}} \right) - 1}{\hat{\lambda}_{1} \left( \bR_{\mathcal{S}} \right)-1}  = \left( \frac{n(n-1)}{2} \cdot \frac{\sum_{k=1}^{n(n-1)/2} \gamma_k^2}{\left(\sum_{k=1}^{n(n-1)/2} \gamma_k \right)^2} \right)^{\frac{1}{2}} \geq 1 \;\; \textrm{for} \;\; n \geq 2
\end{equation}
%\end{small}
because
%\begin{small}
\begin{equation}\label{eq:proof_on_upper_bound_sub1}
\begin{split}
\left(\sum_{k=1}^{n(n-1)/2} \gamma_k \right)^2 &=\left(\sum_{k=1}^{n(n-1)/2} 1 \cdot \gamma_k \right)^2 \\
&\leq  \left(\sum_{k=1}^{n(n-1)/2} 1^2 \right) \left(\sum_{k=1}^{n(n-1)/2} \gamma_k^2 \right)\\
&= \frac{n(n-1)}{2} \left(\sum_{k=1}^{n(n-1)/2} \gamma_k^2 \right)
\end{split} 
\end{equation}
%\end{small}
which results from Cauchy-Schwarz inequality. Therefore, we get 
%\begin{small}
\begin{equation}\label{eq:proof_on_upper_bound_final}
\hat{\lambda}_{1} \left( \bR_{\mathcal{S}} \right) \leq \lambda_{1, \rm{UB}} \left( \bR_{\mathcal{S}} \right) \;\; \textrm{for} \;\; n \geq 2 ,
\end{equation}
%\end{small}
which completes the proof.
\end{proof}
%------------------------------------------------------------------------------------------------------------
\linespread{1.2}
\bibliographystyle{IEEEtran}
\bibliography{IEEEabrv,Refbib_Dynamic}

\end{document}

%% file: input.tex
\usepackage{amsfonts}
\usepackage{times}
\usepackage{graphicx}
\usepackage{latexsym}
\usepackage{dsfont}
\usepackage{amssymb}
\usepackage{amsmath}
\usepackage{cite}
\usepackage{verbatim}
\usepackage{subfigure}
\newtheorem{theorem}{Theorem}

\newtheorem{algorithm}[theorem]{Algorithm}

\newtheorem{proposition}[theorem]{Proposition}

\newenvironment{proof}{ \textbf{Proof:} }{ \hfill $\Box$}

\newcommand{\figref}[1]{{Fig.}~\ref{#1}}

\def\bb0{{\mathbb{0}}}

\def\ba{{\mathbf{a}}}
\def\bb{{\mathbf{b}}}

\def\bff{{\mathbf{f}}}

\def\bm{{\mathbf{m}}}
\def\bn{{\mathbf{n}}}

\def\bs{{\mathbf{s}}}

\def\bx{{\mathbf{x}}}
\def\by{{\mathbf{y}}}

\def\b0{{\mathbf{0}}}

\def\bA{{\mathbf{A}}}

\def\bD{{\mathbf{D}}}

\def\bF{{\mathbf{F}}}

\def\bH{{\mathbf{H}}}
\def\bI{{\mathbf{I}}}

\def\bP{{\mathbf{P}}}

\def\bR{{\mathbf{R}}}

\def\bU{{\mathbf{U}}}
\def\bV{{\mathbf{V}}}
\def\bW{{\mathbf{W}}}
\def\bX{{\mathbf{X}}}

\def\bbE{{\mathbb{E}}}

\def\cA{\mathcal{A}}

\def\sf0{{\mathsf{0}}}

\def\Nt{{N_t}}
\def\Nr{{N_r}}